\documentclass[fleqn,11pt]{article}

\usepackage[utf8]{inputenc}
\usepackage{verbatim}
\usepackage{calc}
\usepackage{amsmath}
\usepackage{amsthm}
\usepackage{amssymb}
\usepackage{mathtools}
\numberwithin{equation}{section}
\numberwithin{figure}{section}
\theoremstyle{plain}
\newtheorem{thm}{\protect\theoremname}
\theoremstyle{definition}
\newtheorem{defn}[thm]{\protect\definitionname}

\providecommand{\definitionname}{Definition}
\providecommand{\theoremname}{Theorem}

\begin{document}
	
	\title{Conformal Yano-Killing tensors for space-times \\ with cosmological
		constant}
	
	\author{Pawe{\l} Czajka\thanks{E--mail: \texttt{pc362945@okwf.fuw.edu.pl}}\; and Jacek Jezierski\thanks{E--mail: \texttt{Jacek.Jezierski@fuw.edu.pl}} \\
		Department of Mathematical Methods in
		Physics, \\ University of Warsaw,
		ul. Pasteura 5, 02-093 Warszawa, Poland}
	
	\maketitle
	\newcommand{\sgn}{\mathop {\rm sgn}\nolimits }
	\global\long\def\dd{{\rm d}}
	\global\long\def\subbb#1{{_{#1}}}
	\global\long\def\suppp#1{{^{#1}}}
	\global\long\def\kropka{\, .}
	\global\long\def\przecinek{\, ,}
	\begin{abstract}
		We present a new method for constructing conformal Yano-Killing
		tensors in five-di\-men\-sio\-nal Anti-de Sitter space-time. The
		found tensors are represented in two different coordinate systems.
		We also discuss, in terms of CYK tensors,
		global charges which are well defined for
		asymptotically (five-dimensional) Anti-de Sitter space-time.
		Additionally in Appendix we present our own derivation of conformal
		Killing one-forms in four-dimensional Anti-de Sitter space-time as an application of the
		Theorem~\ref{cof konforemnej jest konforemne}.
	\end{abstract}
	
	\section{Introduction}
	
	We generalize the construction of the Conformal Yano-Killing
	tensors presented in \cite{chas} to the case of the five-dimensional
	Anti-de Sitter space-time which is intensively explored
	in the context of the AdS/CFT correspondence.
In particular, one can try to generalize formulae (3.11-3.17) from Section 3.2 in \cite{Heller} to the case of conformal tensors. More precisely,
tensor product of two conformal Killing vectorfields $K^\mu$ can be replaced by a symmetric conformal Killing tensor $K^{\mu\nu}$:
 $ \displaystyle K^{\mu}K^{\nu} \langle {\cal O}_{\mu\nu}\rangle \; \longrightarrow \; K^{\mu\nu} \langle {\cal O}_{\mu\nu}\rangle $
but for skew-symmetric tensor $F_{\mu\nu}$ (primary operator) one can consider the expression
$Q_{\mu\nu} \langle F^{\mu\nu} \rangle$ where $Q$ is the conformal Yano-Killing two-form\footnote{Obviously higher rank tensors with more indices are also possible, both symmetric and skew-symmetric.}.
 We also generalize constructions from \cite{chas} to show how conformal Yano-Killing
	tensors can be used to define global gravitational charges in the
	case of the five-dimensional Anti-de Sitter space-time.

	\subsection{Construction of the five-dimensional Anti-de Sitter space-time}
	
	Anti-de Sitter space-time can be constructed in the following way.
	We consider six-dimensional affine space $\tilde{V}$, which is modeled
	on the vector space $V$. The vector space $V$ is equipped with a pseudo-scalar product
	\[
	(v,w)=-v^{0}w^{0}+\sum_{k=1}^{4}v^{k}w^{k}-v^{5}w^{5}\,.
	\]
	Affine space $\tilde{V}$ is naturally a flat manifold. If we choose
	one point in $\tilde{V}$, then it effectively turns our affine space $\tilde{V}$
	into vector space $V$ which is isomorphic to $\mathbb{R}^{6}$ but not canonically.
	Next we consider in the affine space $\tilde{V}$ identified with
	vector space $V$ the locus of the equation
	\begin{equation}
	\label{-lkw} \left(x,x\right)=-l^{2}\, .
	\end{equation}
	Bound $\left(x,x\right)=-l^{2}$ defines five-dimensional submanifold
	of $\tilde{V}$, which as a manifold is the five-dimensional Anti-de Sitter space-time.  In the case of affine space
	$\tilde{V}$ we have canonical isomorphism $\forall_{a\in\tilde{V}}\,T_{a}\tilde{V}\simeq V$.
	This means that $\tilde{V}$ is a (pseudo)-Riemannian manifold. Therefore
	we can pull back the metric to the locus $(x,x)=-l^{2}$ from the ambient
	space $\tilde{V}$. That way our five-dimensional Anti-de Sitter,
	later denoted $\textrm{AdS5}$, gains the structure of (pseudo)-Riemannian
	manifold.
	
	Each linear transformation of $V$ that respects quadratic form $(\cdot,\cdot)$,
	that is a transformation from $SO\left(2,4\right)$ group, preserves
	$\textrm{AdS5}$ as a subset, because if $f\in SO\left(2,4\right)$
	then it follows that $\left(x,x\right)=\left(fx,fx\right)$. The metric
	of $\textrm{AdS5}$ is also preserved, because transformation $f$
	preserves (pseudo)-scalar product of the space $V$, so it also preserves
	(pseudo)-Riemannian metric of the affine space $\tilde{V}$. That
	means that it also preserves metric induced on $\textrm{AdS5}$, which
	is submanifold of $\tilde{V}$. This shows that
	 $SO\left(2,4\right)$ is a subgroup of the isometry group of $\textrm{AdS5}$ (it is in fact the whole isometry group).
	
	\section{Description of the metric submanifold}
	\label{submanifoldy}
	
	Inspired by this example, let's consider general situation. We have
	a pair $\left(N,g\right)$, where $g$ is a metric of manifold $N$.
	We also have submanifold $M$ of codimension 1 with metric $\tilde{g}$
	induced from $N$. Let $\overset{N}{\nabla}_{X}Y$ be the Levi-Civita
	derivative of the vector field $Y$ tangent to $N$ with respect to
	the field $X$, which is also tangent to $N$. If fields $X,Y$ are tangent to $M$, that is $X,Y\in\Gamma\left(TM\right)$,
	 we will also write $\overset{N}{\nabla}_{X}Y$ understanding,
	that in this notation fields $X,Y$ are substituted by their arbitrary local
	extensions. The result on $M$ does not depend on the choice of those
	extensions. Next we denote as $\overset{M}{\nabla}_{X}Y$ the derivative
	of the field $Y$ tangent to $M$ with respect to the field $X$ also
	tangent to $M$ with respect to the metric connection on $M$. In this
	notation we have $X,Y\in\Gamma\left(TM\right)\Rightarrow\overset{N}{\nabla}_{X}Y=\overset{M}{\nabla}_{X}Y+\tilde{K}\left(X,Y\right)$,
	where $\tilde{K}\left(X,Y\right)$ is the form of external curvature.
	It is also called the second fundamental form. It is known that $\tilde{K}\left(X,Y\right)\perp TM$
	and $\tilde{K}\left(X,Y\right)=\tilde{K}\left(Y,X\right)$.
	
	If we choose local coordinate system $x_{1},\dots , x_{n+1}$ on some
	open subset of $N$ which satisfies $\emptyset\neq\left\{ p\in N\mid x\subbb{n+1}=0\right\} \subset M$,
	then the collection of functions $x_{1},\dots , x_{n}$ is a local coordinate
	system on $M$. Let Latin indices go from $1$ to $n$, whereas Greek
	letters go from $1$ to $n+1$. We also choose normal field $n$ defined
	on $M$ such that $n\in\Gamma\left(TN\right)$, $n\perp TM$, and
	$\left(n,n\right)=\pm1$. Here we choose the sign depending on the type
	of the surface (null surfaces are not considered here). We can now write $\tilde K (X,Y)=K(X,Y) n$ which defines  $K$ as a symmetric tensor of rank 2. We also use
	convenient notation, in which $v^{\mu}{_{|\nu}}=\overset{N}{\nabla}_{\nu}v^{\mu}$,
	and $v^{a}{_{;b}}=\overset{M}{\nabla}_{b}v^{a}$.
	\begin{thm}
		Let $\omega$ be a one-form on the manifold $N$. Then\label{pochodnajednoformyyy}
		\[
		\omega_{b|a}=\omega_{b;a}-K_{ab}\omega_{\mu}n^{\mu}\,.
		\]
	\end{thm}

	The proof of this theorem  can be found in appendix B.
	
	\begin{thm}
		\label{krzywizna zewnetrzna}The external curvature form $K_{ab}$
		satisfies equation $K=-\frac{1}{2}\mathcal{L}_{n}g$.
	\end{thm}
	The proof of this theorem  can be found in  appendix B.

	\begin{thm}
		\label{thm:ads jest dobrze zanurzony}In the case $N=\tilde{V}$ and
		$M=\textrm{AdS5}$ we have $K=C\tilde{g}$, where $C$ is some real
		function on $M$.
	\end{thm}
	
	\begin{proof}
		We have the identity $K=-\frac{1}{2}\mathcal{L}_{n}g$. On the vector
		space with pseudo-scalar product we can always choose coordinates
		$\left(r,\phi_{i}\right)$, where $r(p)=\sqrt{\mid g(p-0,p-0) \mid},\quad p\in N$ ($0$ here is an arbitrbitrally chosen point in $N$) is the distance from the zero
		vector, whereas $\phi_{i}$ are some angles that are coordinate system
		of the pseudo-sphere of constant $r$. Additionally we can choose
		 coordinates $\phi_i$ in such a way that the metric has the form
		\begin{equation}
		g=\pm dr^{2}\pm r^{2}\hat{g}\,,
		\end{equation}
		where $\hat{g}$ is a metric of the unit pseudo-sphere parametrized by $\phi_i$ and does not
		depend on the coordinate $r$. Lie derivative along the field $n=\partial_{r}$
		of the metric $g$ is of course proportional to $\hat{g}$ which is
		proportional to the induced metric on pseudo-sphere.
	\end{proof}
	
	So far the results are repeated to fix the notation -- the subject of Theorems 1-3 is well-established. 
	\section{Pulling back conformal tensors to submanifolds}
	
	Theorem \ref{thm:ads jest dobrze zanurzony} suggests to restrict
	our considerations to the case when $K\sim\tilde{g}.$ From now on
	we assume that this condition holds. Now we can prove the following
	theorem.
	
	\begin{thm}
		\label{cof konforemnej jest konforemne} If $k$ is a conformal Killing
		one-form on $N$ then its pullback to $M$ is a conformal Killing one-form
		on $M$.
	\end{thm}
	\begin{proof}
		Let's compute
		\[
		k_{\left(a|b\right)} + K_{ab}n^{\mu}k_{\mu} = k_{\left(a;b\right)} \; (=Ag_{ab})\,.
		\]
		where $A$ is a function. We see that $k_{\left(a;b\right)}\sim\tilde{g}_{ab}$
		because both terms above are proportional to the metric tensor $g_{ab}$ on $M$.
		Let us notice that in this case the restriction of a Killing one-form,
		that is one-form such that $k_{\left(a|b\right)}=0$, in some cases
		won't be a Killing one-form on $M$ but only conformal one.
	\end{proof}
	
	\begin{thm}
		\label{thm:pochodna dwuformy}We have the following identity for computing
		the covariant derivative of the two-form $Q$ on the manifold $N$:
		\[ Q_{ac|b}=Q_{ac;b}-K_{ab}Q_{\nu c}n^{\nu}-K_{bc}Q_{a\mu}n^{\mu} \, . \]
	\end{thm}
	\begin{proof}
		We contract the two-form $Q$ with arbitrary vector field $v$ tangent
		to $M$. We can compute the derivative of the resulting one-form using
		the formula \ref{eq:poch kow kowektor}.
		\begin{equation}
		(Q_{a\mu}v^{\mu})_{|b}=(Q_{ac}v^{c})_{;b}-K_{ab}Q_{\nu\mu}v^{\mu}n^{\nu}
		=Q_{ac;b}v^{c}+Q_{ac}v^{c}{_{;b}}-K_{ab}Q_{\nu\mu}v^{\mu}n^{\nu}\, .
		\end{equation}
		On the other hand
		\begin{equation}
		(Q_{a\mu}v^{\mu})_{|b}=Q_{ac|b}v^{c}+Q_{ac}v_{;b}^{c}+Q_{a\mu}K_{bc}v^{c}n^{\mu}\,.
		\end{equation}
		Comparison of the two sides of equations leads to the conclusion that
		\begin{equation}
		Q_{ac|b}v^{c}=Q_{ac;b}v^{c}-K_{ab}Q_{\nu c}n^{\nu}v^{c}-Q_{a\mu}K_{bc}v^{c}n^{\mu}\,,\label{eq:poch kow dwuforma}
		\end{equation}
		so
		\begin{equation}
		Q_{ac|b}=Q_{ac;b}-K_{ab}Q_{\nu c}n^{\nu}-K_{bc}Q_{a\mu}n^{\mu}\,.
		\end{equation}
		It is easy to generalize this identity to arbitrary $n$-forms. For a three-form
		the identity is given by theorem \ref{thm:pochodna kowariantna trojformy}. The identity for
		a two-form can be written as
		\begin{equation}
		Q_{ab|c}=Q_{ab;c}-q_{a}K_{cb}+q_{b}K_{ac}\,,
		\end{equation}
		where $q_{a}=Q_{a\mu}n^{\mu}$.
	\end{proof}
	\begin{defn}
		The two-form $Q$ satisfying equation $Q_{\alpha(\beta;\gamma)}=0$ is
		called Yano-Killing tensor.
	\end{defn}
	
	\begin{thm}
		If $\tilde{Q}$ is a Yano-Killing tensor on the manifold $N$, then
		its pullback to sub-manifold $M$ denoted by $Q$ satisfies
		\[ Q_{ab;c}+Q_{ac;b}=2\tilde{q}_{a}g_{cb}-\tilde{q}_{b}g_{ac}-\tilde{q}_{c}g_{ab} \, , \]
		where $\tilde{q}$ is a certain one-form\footnote{$\tilde{q}$ is obviously related to the divergence of $Q$
			by contraction of the indices in the above equation.}.
	\end{thm}
	\begin{proof}
		Let's check what equation is satisfied by the pullback of the
		form $Q$ to $M$. We have
		\begin{equation}
		0=Q_{ab|c}+Q_{ac|b}=Q_{ab;c}+Q_{ac;b}-2q_{a}K_{cb}+q_{b}K_{ac}+q_{c}K_{ab}\,.
		\end{equation}
		So it turns out that pullback of the Yano-Killing tensor is satisfying
		a bit different equation then Yano-Killing tensors. This equation
		looks like this
		\begin{equation}
		Q_{ab;c}+Q_{ac;b}=2\tilde{q}_{a}g_{cb}-\tilde{q}_{b}g_{ac}-\tilde{q}_{c}g_{ab}\,,\label{cyk rownanie}
		\end{equation}
		where $\tilde{q}_{c}$ is a certain one-form.
	\end{proof}
	Last theorem suggests the following definition:
	\begin{defn}
		\label{cykdef}
		If $M$ is a Riemannian manifold and $Q$ is the two-form satisfying equation
		\[  Q_{ab;c}+Q_{ac;b}=2\tilde{q}_{a}g_{cb}-\tilde{q}_{b}g_{ac}-\tilde{q}_{c}g_{ab} \]
		for some one-form $\tilde{q}$, then $Q$ is called conformal Yano-Killing
		tensor. We will often use abbreviation CYK tensor for conformal Yano-Killing
		tensor.
		
	\end{defn}
	
	We decided to find CYK tensors on five-dimensional Anti-de Sitter space-time.
	In the ambient vector space $V$ that surrounds five-dimensional Anti-de Sitter space-time
	it is easy to find some CYK tensors. We can just choose two-forms which
	have constant coefficients in Cartesian coordinates. This way we can
	obtain $15$ CYK tensors on Anti-de Sitter space-time. However it is known that
	there are $35$ linearly independent CYK tensors on this space-time.
	We will now present a way to find the remaining $20$ CYK tensors.
	\begin{thm}
		\label{thm:pochodna kowariantna trojformy}Analogously to the equation
		\ref{eq:poch kow dwuforma} it can be proved that the covariant derivative
		of three-form $T_{\alpha\beta\gamma}$ looks like this
		\begin{equation}
		T_{abc|d}=T_{abc;d}-Q_{bc}K_{ad}+Q_{ac}K_{bd}-Q_{ab}K_{cd}\,,
		\end{equation}
		where $Q_{ab}=T_{ab\mu}n^{\mu}$.
	\end{thm}
	
	\begin{proof}
		Analogous to the proof of the theorem \ref{thm:pochodna dwuformy}.
	\end{proof}
	\begin{thm}
		If a three-form $T$ on the manifold $N$ satisfies the equation $T_{\alpha\beta\left(\gamma|\delta\right)}=0$,
		then its pullback to the manifold $M$ satisfies the equation
		\[ 2T_{ab\left(c;d\right)} = 2Q_{ab}g_{cd}-Q_{ac}g_{bd}-Q_{ad}g_{bc}+Q_{bc}g_{ad}+Q_{bd}g_{ac} \]
		with a certain two-form $Q$.
	\end{thm}
	
	\begin{proof}
		\begin{equation}
		2T_{ab\left(c|d\right)} =
		2T_{ab\left(c;d\right)}-2Q_{ab}K_{cd}+Q_{ac}K_{bd}+Q_{ad}K_{bc}-Q_{bc}K_{ad}-Q_{bd}K_{ac}\,.
		\end{equation}
		From this equation it follows that if $T$ satisfies $T_{\alpha\beta\left(\gamma|\delta\right)}=0$,
		then pullback $T$ to $M$ satisfies
		\begin{equation}
		2T_{ab\left(c;d\right)} =
		2Q_{ab}g_{cd}-Q_{ac}g_{bd}-Q_{ad}g_{bc}+Q_{bc}g_{ad}+Q_{bd}g_{ac}\,.\label{eq:cyk trojforma}
		\end{equation}
	\end{proof}
	
	\begin{defn}
		The three-form $T_{abc}$ satisfying equation \ref{eq:cyk trojforma}
		is called CYK three-form.
	\end{defn}

	\section{Five-dimensional case}
	
	In this section we don't take into considerations the surrounding
	manifold $N$. Additionally, all tensor fields are defined on $M$
	and $\dim M=5$. In this case we have the following well known theorems.
	\begin{thm}
		Hodge dual of the CYK three-form is a CYK tensor.
		\label{dualofcyktensor}
	\end{thm}
	The proof can be found in the appendix B but a general case is also given in
Proposition 3.2 in arXiv:1104.0852.

\begin{thm}
	If $k$ is conformal Killing one-form for the metric $g$ and $\Omega^2$ is a positive smooth function then $\Omega^2 k$ is a conformal Killing one-form for the metric $\Omega^2 g$. \label{confkilloneconf}
\end{thm}
\begin{proof}
	Let's denote as $X$ the vector field associated with one-form $k$ as follows $X^i=g^{ij}k_j$ (in this proof indices $i,j$ go through all functions from our coordinate system). It is known that conformal Killing equation $\nabla_{(a}k_{b)}=\lambda^\prime g_{ab}$, where $\lambda^\prime$ is some function is equivalent to the equation $\mathcal L_X g=\lambda g$. Now lets compute $\mathcal L_X(\Omega^2 g)=\mathcal L_X(\Omega^2)g+\Omega^2 \mathcal L_Xg=(\frac{\mathcal L_X \Omega^2}{\Omega^2}+\lambda)(\Omega^2 g)$. This shows, that vector field $X$ is related to conformal Killing one-form $b$ for the metric $\Omega^2 g$. This one-form is equal to $b_i=\Omega^2 g_{ij}X^j=\Omega^2k_i$.
\end{proof}	
	
	\begin{thm}
		If $Q$ is a CYK tensor for the metric $g$, then $\Omega^{3}Q$ is
		a CYK tensor for the metric $\Omega^{2}g$.\label{thm:konforemnosc cyk tensorow}
	\end{thm}
The proof can be found in \cite{cykkerr}.

	\subsection{Construction of conserved charges in asymptotically
		Anti-de Sitter space-times}
	\begin{defn}
		Tensor field $W$ is called spin-2 field if it satisfies
		\[ W_{\alpha\beta\gamma\delta}=W_{\gamma\delta\alpha\beta}
		=W_{[\alpha\beta][\gamma\delta]},\,\,W_{\alpha[\beta\gamma\delta]}=0,
		\,\,W_{\,\,\beta\alpha\delta}^{\alpha}=0,\,\,\nabla_{[\lambda}W_{\mu\nu]\alpha\beta}=0 \, . \]
	\end{defn}
	
	An example of the spin-2 field is Weyl tensor. In case of this tensor
	we also know that conformal transformations do not change Weyl components
	$W^{\alpha}{_{\beta\gamma\delta}}$ .
	
	Next theorem enables one to define conserved charges on
	space-times that are asymptotically similar to Anti-de Sitter space-time.
	\begin{thm}
		\label{thm:T zamkni=00003D00003D000119ta}If $Q$ is a CYK tensor
		and $W$ is a spin-2 field, then the three-form \\ $T_{\alpha\beta\gamma}=\frac{1}{2}\epsilon_{\alpha\beta\gamma}{^{\delta\sigma}}W_{\delta\sigma\mu\nu}Q^{\mu\nu}$
		is closed.
	\end{thm}
	
	\begin{proof}
		Let us define
		\begin{equation}
		F_{\mu\nu}=W_{\mu\nu\lambda\kappa}Q^{\lambda\kappa}\,,
		\end{equation}
		where $Q$ is a certain CYK tensor. We will show that
		\begin{equation}
		F^{\mu\nu}{_{;\nu}}=\frac{2}{3}W^{\mu\nu\alpha\beta}\mathcal{Q}_{\alpha\beta\nu}\,,\label{div F}
		\end{equation}
		where
		\begin{equation}
		 \mathcal{Q}_{\lambda\kappa\sigma}\left(Q,g\right)=Q_{\lambda\kappa;\sigma}+Q_{\sigma\kappa;\lambda}-\frac{1}{2}\left(g_{\lambda\sigma}Q^{\nu}{_{\kappa;\nu}}+g_{\kappa(\lambda}Q_{\sigma)}{^{\mu}}{_{;\mu}}\right)\,,\label{cyk rownanie doklandiej}
		\end{equation}
		so $\mathcal{Q}=0$ if $Q$ is a CYK tensor (it follows from contraction
		of a pair of indices in equation \ref{cyk rownanie}). We can prove
		the equation \ref{div F} in the following way
		\begin{eqnarray*}
			W^{\mu\nu\alpha\beta}\mathcal{Q}_{\alpha\beta\nu} &= &
			W^{\mu\nu\alpha\beta}\left(Q_{\alpha\beta;\nu}+Q_{\nu\beta;\alpha}\right)=
			\left(W^{\mu\nu\alpha\beta}+W^{\mu\alpha\nu\beta}\right)Q_{\alpha\beta;\nu}\\
			&=& \left(W^{\mu\nu\alpha\beta}+ \frac{1}{2}W^{\mu\alpha\nu\beta}-\frac{1}{2}W^{\mu\beta\nu\alpha}\right)Q_{\alpha\beta;\nu}\\
			&=& \left[\frac{3}{2}W^{\mu\nu\alpha\beta}-\frac{1}{2}\left(W^{\mu\nu\alpha\beta} +W^{\mu\alpha\beta\nu}+W^{\mu\beta\nu\alpha}\right)\right] Q_{\alpha\beta;\nu}\\
			&=& \frac{3}{2}W^{\mu\nu\alpha\beta}Q_{\alpha\beta;\nu}
		\end{eqnarray*}
		For this reason
		\begin{equation}
		\nabla_{\nu}F^{\mu\nu}=\nabla_{\nu}\left(W^{\mu\nu\alpha\beta}Q_{\alpha\beta}\right)
		=\left(\nabla_{\nu}W^{\mu\nu\alpha\beta}\right)Q_{\alpha\beta}
		+W^{\mu\nu\alpha\beta}\nabla_{\nu}Q_{\alpha\beta} \, .
		\end{equation}
		Let us notice that if we contract indices $\mu$ and $\alpha$ in
		the equation
		\begin{equation}
		\nabla_{\lambda}W_{\mu\nu\alpha\beta}+\nabla_{\mu}W_{\nu\lambda\alpha\beta}
		+\nabla_{\nu}W_{\lambda\mu\alpha\beta}=0 \, ,
		\end{equation}
		then we will end up with
		\begin{equation}
		\nabla_{\alpha}W_{\lambda\mu}{^{\alpha}}{_{\beta}}=0\,,
		\end{equation}
		and finally
		\begin{equation}
		\nabla_{\nu}F^{\mu\nu}=W^{\mu\nu\alpha\beta}\mathcal{Q}_{\alpha\beta;\nu}=\frac 2 3 W^{\mu\nu\alpha\beta}\mathcal{Q}_{\alpha\beta\nu}\,.
		\end{equation}
		In our case $Q$ is a CYK tensor, so
		\begin{equation}
		\nabla_{\nu}F^{\mu\nu}=0\,.
		\end{equation}
		We can always express $F$ as $F=\ast T$. We can use identity $\ast\ast F=(-1)^s F$, 
  $(-1)^s:=\sgn\det(g)$ (valid for two-forms $F$ in five-dimensional pseudo-Riemannian space) to obtain
		\begin{equation}
		T=\left(-1\right)^{s}\ast\!F\,.\label{eq:trojforma}
		\end{equation}
		
		Next we have
		\begin{equation}
		\nabla_{\nu}F^{\nu\mu}=
		\frac{1}{6}\nabla_{\nu}\left(\epsilon^{\nu\mu\alpha\beta\gamma}T_{\alpha\beta\gamma}\right)=
		\frac{1}{6}\epsilon^{\nu\mu\alpha\beta\gamma}\nabla_{\nu}T_{\alpha\beta\gamma}
		=\frac{1}{6}\epsilon^{\nu\mu\alpha\beta\gamma}\partial_{[\nu}T_{\alpha\beta\gamma]} \, .
		\end{equation}
		We can change covariant derivatives to partial derivatives
		because Christoffel symbols are symmetric in their
		lower indices. This shows that $\partial_{[\nu}T_{\alpha\beta\gamma]}=0$,
		and therefore ${\rm d}T=0$.
	\end{proof}
\subsection{A way to construct a quasi-local charge}
	In the theorem \ref{thm:T zamkni=00003D00003D000119ta} we have found a way to obtain a closed three-form. For space-time
	which is asymptotically similar to Anti-de Sitter space-time (it
	means that there exist coordinate system in which metric is similar
	to Anti-de Sitter metric close to infinity, see \cite{key-2}) we can construct
	asymptotic CYK tensor which asymptotically satisfy CYK equation from the definition \ref{cykdef}.
	In this way, if we also have spin-2 field on our asymptotically
	Anti-de Sitter space-time (for instance Weyl tensor), then we can construct asymptotically closed
	three-form. We can now consider the slice of constant time. We
	integrate our three-form on a large three-dimensional sphere belonging to this
	slice and located in the asymptotic region. We will end up with a quantity that asymptotically doesn't depend
	on the size of this sphere or rather approaches (possibly finite) limit at infinity. This way we obtain some quasi-local charge. It
	turns out, that if we change the metric $g$ (by conformal rescaling) to the $\Omega^{2}g$, and
	if our spin-2 field $W$ is chosen to be Weyl tensor, then the corresponding the form $T$
	transforms to $\Omega^{2}T$. \label{konforemnosc}
	
	
	\section{CYK tensors in coordinate systems}
	
	It follows from our past considerations that on $\textrm{AdS5}$ one
	can find CYK tensors as pullbacks of constant two-forms on $\tilde{V}$
	and as Hodge duals of pullbacks of constant three-forms on the surrounding
	space $\tilde{V}$.
	
	We will use  the convention that indices $a,b,c$ go from $1$ to $3$,
	indices $i,j,k$ go from $1$ to $4$, indices $\mu, \nu, \lambda$ go
	from $0$ to $4$, and indices $A,B,C$ go from $0$ to $5$.
	
	\subsection{Poincar{\`e} coordinate system}
	
	Let us consider a parametrization of Anti-de Sitter space-time with coordinates
	$t$, $x^{1}$, $x^{2}$, $x^{3}$, $y$. This means that $t$ has index $0$, $x^{1}$
	has index $1$ and so on. Quantity $l$ is a parameter that is related
	to the size of Anti-de Sitter. This parameter is also a part of equation  $\left(X,X\right)=-l^{2}$
	(this is equation \ref{-lkw})
	defining $\textrm{AdS5}$. 
	In these coordinates
	we have
	\begin{align*}
	X^{0} & =\frac{1}{2y}(y^{2}+l^{2}+ \| \bar{x}\|^{2}-t^{2})\,,\\
	X^{a} & =\frac{x^{a}}{y}l \, , \quad a\in\{1,2,3\}\,,\\
	X^{4} & =\frac{1}{2y}(y^{2}-l^{2}+\|\bar{x}\|^{2}-t^{2})\,,\\
	X^{5} & =\frac{t}{y}l\,,
	\end{align*}
	where metric on the space $\tilde{V}$ is equal to
	\begin{equation}
	{\rm d}s^{2}=-\left({\rm d}X^{0}\right)^{2}+\sum_{k=1}^{4}\left({\rm d}X^{k}\right)^{2}-\left({\rm d}X^{5}\right)^{2}\,.
	\end{equation}
	$\textrm{AdS5}$ is a locus
	\begin{equation}
	-l^{2}=-\left(X^{0}\right)^{2}+\sum_{k=1}^{4}\left(X^{k}\right)^{2}-\left(X^{5}\right)^{2}\,.
	\end{equation}
	It turns out, that in these coordinates the induced metric is conformally
	flat and equal to
	\begin{equation}
	{\rm d}s^{2}=\frac{l^{2}}{y^{2}}(-{\rm d}t^{2}+{\rm d}y^{2}+{\rm d}\bar{x}^{2})\,,
	\end{equation}
	where $\|\bar{x}\|^{2}=\left(x^{1}\right)^{2}+\left(x^{2}\right)^{2}+\left(x^{3}\right)^{2}$,
	whereas ${\displaystyle {\rm d}\bar{x}^{2}=\sum_{a=1}^{3}\left({\rm d}x^{a}\right)^{2}}$.
	For this reason, if we choose conformal factor $\Omega=\frac{l}{y}$ in the theorem \ref{thm:konforemnosc cyk tensorow}, then
	we see, that CYK tensors on $\textrm{AdS5}$ 
	divided by $\Omega^{3}$
	are CYK tensors on the five-dimensional Minkowski space-time. Let's denote
	the pullbacks of constant two-forms as 
	\begin{equation}
	C_{AB}:=\mathfrak{i}^{\ast}\left({\rm d}X^{A}\wedge{\rm d}X^{B}\right)\,,
	\end{equation}
	and the Hodge
	duals of pullbacks of constant three-forms as 
	\begin{equation}
	H_{ABC}:=\ast\mathfrak{i}^{\ast}\left({\rm d}X^{A}\wedge{\rm d}X^{B}\wedge{\rm d}X^{C}\right)\,,
	\end{equation}
	where $\mathfrak{i}$ is an immersion of the Anti-de Sitter space-time
	into $6$ dimensional ambient  vector space $V$.
	
	Let's adopt the following notation
	: 
	$D=x^{a}{\rm d}x^{a}+y{\rm d}y-t{\rm d}t$, $\mathcal{D}=x^{a}{\rm d}x^{a}+y{\rm d}y$,
	$\tau_{a}={\rm d}x^{a}$, $\tau_{4}={\rm d}y$, $\mathcal{K}_{a}=x^{a}\mathcal{D}-\frac{1}{2}\left(\bar{x}^{2}+y^{2}\right)\tau_{a}$,
	$\mathcal{K}_{4}=y\mathcal{D}-\frac{1}{2}\left(\bar{x}^{2}+y^{2}\right)\tau_{4}$,
	$\mathcal{L}_{ab}=x^{a}{\rm d}x^{b}-x^{b}{\rm d}x^{a}$, $\mathcal{L}_{a4}=x^{a}{\rm d}y-y{\rm d}x^{a}$.
	
	We calculated the tensors $C_{AB}$ and
	$H_{ABC}$ in Mathematica.
	We can express them in the above notation. Let's consider an
	array of numbers $\epsilon^{ijkl}$ which gives to the collection
	of indices $i,j,k,l\in\left\{ 1,2,3,4\right\} $ the sign of the permutation associated with them or zero (if this collection of indices is not a permutation,
	then the result is 0). We also use here the old summation
	convention. This means that we contract the same indices even if they
	are on the same level.
	
	\[
	C_{0,4}=\Omega^{3}\frac{1}{l}\left[\tau_{4}\wedge\left(-D\right)\right]
	\]
	\[
	C_{0,5}=\Omega^{3}\frac{1}{l^{2}}\left[-{\rm d}t\wedge\mathcal{K}_{4}+\frac{1}{2}l^{2}dt\wedge\tau_{4}+t\tau_{4}\wedge\mathcal{D}-\frac{1}{2}t^{2}\tau_{4}\wedge{\rm d}t\right]
	\]
	\[
	C_{4,5}=\Omega^{3}\frac{1}{l^{2}}\left[{\rm d}t\wedge\left(-\mathcal{K}_{4}\right)-\frac{1}{2}l^{2}{\rm d}t\wedge\tau_{4}+t\tau_{4}\wedge\mathcal{D}-\frac{1}{2}t^{2}\tau_{4}\wedge{\rm d}t\right]
	\]
	\[
	C_{0,a}=\Omega^{3}\frac{1}{l^{2}}\left[\mathcal{L}_{a,4}\wedge D+\frac{1}{2}(D,D)\tau_{a}\wedge\tau_{4}-\frac{1}{2}l^{2}\tau_{4}\wedge\tau_{a}\right]
	\]
	\[
	 C_{a,4}=\Omega^{3}\frac{1}{l^{2}}\left[-\frac{1}{2}l^{2}\tau_{4}\wedge\tau_{a}+D\wedge\mathcal{L}_{a,4}+\frac{1}{2}D^{2}\tau_{4}\wedge\tau_{a}\right]
	\]
	\[
	C_{a,5}=\Omega^{3}\frac{1}{l}\left[{\rm d}y\wedge\left(t{\rm d}x^{a}-x^{a}{\rm d}t\right)+y{\rm d}x^{a}\wedge{\rm d}t)\right]
	\]
	\[
	C_{ab}=\Omega^{3}\frac{1}{l}\left[{\rm d}y\wedge(\mathcal{L}\subbb{ba})+y{\rm d}x^{a}\wedge{\rm d}x^{b}\right]
	\]
	\[
	H_{0,4,5}=\Omega^{3}\frac{\sgn y}{l}\left[\frac{1}{2}\epsilon\suppp{ijk4}x_i\tau\subbb j\wedge\tau\subbb k\right]
	\]
	\[
	H_{0,d,4}=\Omega^{3}\frac{\sgn y}{l}\left[-\frac{1}{2}t\epsilon\suppp{dij4}\tau\subbb i\wedge\tau\subbb j-\frac{1}{2}\epsilon^{dab4}\mathcal{L}_{a,b}\wedge{\rm d}t\right]
	\]
	\[
	H_{0,d,5}=\frac{1}{2}\Omega^{3}\frac{\sgn y}{l^{2}}\left[\epsilon^{dab4}\left(-\mathcal{L}_{a,b}\wedge D-\frac{1}{2}D^{2}\tau_{a}\wedge\tau_{b}+\frac{1}{2}l^{2}\tau_{a}\wedge\tau_{b}\right)\right]
	\]
	\[
	H_{d,4,5}=\frac{1}{2}\Omega^{3}\frac{\sgn y}{l^{2}}\left[\epsilon^{dab4}\left(\mathcal{L}_{a,b}\wedge D+\frac{1}{2}D^{2}\tau_{a}\wedge\tau_{b}+\frac{1}{2}l^{2}\tau_{a}\wedge\tau_{b}\right)\right]
	\]
	\[
	H_{0,d}=\Omega^{3}\frac{\sgn y}{l^{2}}\left[t\mathcal{D}\wedge{\rm d}x^{d}+\left(-\mathcal{K}\subbb d+\frac{1}{2}\left(-l^{2}+t^{2}\right){\rm d}x^{d}\right)\wedge{\rm d}t\right]
	\]
	\[
	H_{d,4}=\Omega^{3}\frac{\sgn y}{l^{2}}\left[t\mathcal{D}\wedge{\rm d}x^{d}+\left(-\mathcal{K}\subbb d+\frac{1}{2}\left(l^{2}+t^{2}\right){\rm d}x^{d}\right)\wedge{\rm d}t\right]
	\]
	\[
	H_{d,5}=\Omega^{3}\frac{\sgn y}{l}\left[\dd x\suppp d\wedge D\right]
	\]
	\[
	H_{1,2,3}=\Omega^{3}\frac{\sgn y}{l}\left[{\rm d}t\wedge\mathcal{D}\right]
	\]
	All those CYK tensors are written in the form $\alpha\left[\beta\right]$
	where $\alpha$ consists of conformal coefficient multiplied by locally
	constant terms like $\sgn y$ and $l$. Theorem \ref{thm:konforemnosc cyk tensorow}
	ensures us that that $\beta$ is a CYK tensor for the metric $\Omega^{-2}\left(\frac{l^{2}}{z^{2}}(-{\rm {\rm d}}t^{2}+{\rm d}y^{2}+{\rm d}\bar{x}^{2})\right)=-{\rm d}t^{2}+{\rm d}y^{2}+{\rm d}\bar{x}^{2}$
	which is equal to the five-dimensional Minkowski metric.
	
	\subsection{Spherical coordinate system}
	
	Coordinates $y,x^{1},x^{2},x^{3},t$ are not convenient because we
	are interested in the form of the CYK tensors on the conformal
	verge -- scri. That means that we want to set $y$ equal to 0. The scri
	of $\textrm{AdS5}$ has the topology of $\mathbb{R}\times S^{3}$,
	however in those coordinates the sphere is parameterized inconveniently.
	For this reason we consider the following parametrization
	\begin{align}
	X_{0} & =\sqrt{l^{2}+r^{2}}\cos\left(\frac{t}{l}\right)\, ,\\
	X_{k} & =rn^{k} \, ,\\
	X_{5} & =\sqrt{l^{2}+r^{2}}\sin\left(\frac{t}{l}\right)\, ,
	\end{align}
	where $\sum_{i=1}^{4}\left(n^{i}\right)^{2}=1$. That means that $n^{k}$
	can be parameterized with 3 angles. We also introduce the coordinate
	$z$ which replaces the coordinate $r$
	\begin{equation}
	r=l\frac{1-z^{2}}{2z}\, , \quad z\in[0,1] \, . \label{r od z w ads}
	\end{equation}
	The choice of this coordinate is justified by the observation that
	it solves the equation
	\begin{equation}
	\left\Vert \frac{l{\rm d}z}{z}\right\Vert ^{2}=1\,.\label{eq:  eikona=00003D00003D00003D000142}
	\end{equation}
	This means that in the conformally equivalent  metric $\frac{z^{2}}{l^{2}}g$ the coordinate $z$
	is easily related  to the distance from the center of {\rm AdS5}.
	In those coordinates Anti-de Sitter metric equals
	\begin{equation}
	g=\frac{l^{2}}{z^{2}}\left({\rm d}z^{2}-\left(\frac{1+z^{2}}{2}\right)^{2}\frac{{\rm d}t^{2}}{l^{2}}+\left(\frac{1-z^{2}}{2}\right)^{2}{\rm d}\Omega_{3}\right)\,.\label{ads w sferycznych z n}
	\end{equation}
	We can divide it by the conformal factor $\frac{l^{2}}{z^{2}}=\Omega^{2}$,
	and then go to the conformal scri $z=0$. Scri is a manifold
	that has metric defined up to conformal rescaling (this ambiguity
	arises because we could divide the Anti-de Sitter metric by an arbitrary
	conformal factor). In those coordinates scri $\mathbb{R}\times S^{3}$
	is conveniently parameterized because $t$ parameterizes $\mathbb{R}$,
	whereas $(n_{k})_{k\in\{1,2,3,4\}}$ parameterize $S^{3}$.
	
	During
	calculations involving Hodge dual we used the following coordinates
	\begin{align}
	X_{0} & =l\sqrt{1+\frac{\left\Vert {p}\right\Vert ^{2}}{l^{2}}}\cos\left(\frac{t}{l}\right)\, , \\
	X_{k} & =p^{k} \, , \quad k\in\{1,2,3,4\}\, , \\
	X_{5} & =l\sqrt{1+\frac{\left\Vert {p}\right\Vert ^{2}}{l^{2}}}\sin\left(\frac{t}{l}\right)\,,
	\end{align}
	where we denoted ${\displaystyle \|p\|^{2}:=\sum_{i=1}^{4}\left(p{^{i}}\right)^{2}}$,
	and then we expressed the resulting CYK tensors through functions $t$, $z$, $n^{1}$, $n^{2}$, $n^{3}$, $n^{4}$ and their exterior derivatives.
	
	This way we obtained CYK tensors on Anti-de Sitter space-time. We adhered to
	our convention that Latin indices go from 1 to 4. We carried out calculations
	in Mathematica.

	\hspace*{-1cm} \begin{eqnarray}
	C_{0,k}&=&{\rm d}t\wedge\left(\frac{l\left(z^{2}+1\right)^{2}\sin\left(\frac{t}{l}\right)n^{k}}{4z^{3}}{\rm d}z+\frac{l\left(z^{4}-1\right)\sin\left(\frac{t}{l}\right)}{4z^{2}}{\rm d}n^{k}\right)-\frac{l^{2}\left(z^{2}-1\right)^{2}\cos\left(\frac{t}{l}\right)}{4z^{3}}{\rm d}z\wedge{\rm d}n^{k}\nonumber \\	
	C_{0,5}&=&-\frac{l\left(z^{4}-1\right)}{4z^{3}}{\rm d}t\wedge{\rm d}z\label{eq: c05}\\
	C_{i,j}&=&\frac{l^{2}\left(z^{4}-1\right)}{4z^{3}}{\rm d}z\wedge\left(n^{i}{\rm d}n^{j}-n^{j}{\rm d}n^{i}\right)+\frac{l^{2}\left(z^{2}-1\right)^{2}}{4z^{2}}{\rm d}n^{i}\wedge{\rm d}n^{j}\nonumber\\
	C_{k,5}&=&{\rm d}t\wedge\left(\frac{l\left(z^{2}+1\right)^{2}\cos\left(\frac{t}{l}\right)n_{k}}{4z^{3}}{\rm d}z+\frac{l\left(z^{4}-1\right)\cos\left(\frac{t}{l}\right)}{4z^{2}}{\rm d}n^{k}\right)+\frac{l^{2}\left(z^{2}-1\right)^{2}\sin\left(\frac{t}{l}\right)}{4z^{3}}{\rm d}z\wedge{\rm d}n^{k}\nonumber
\end{eqnarray}
\hspace*{-1cm} \begin{eqnarray}
	H_{0,i,j}&= &
	\frac{1}{2}\epsilon_{ijkl}\Bigr[\left(\frac{l\left(z^{2}-1\right)^{2}
		\left(z^{2}+1\right)\cos\left(\frac{t}{l}\right)}{8z^{3}}{\rm d}t-\frac{l^{2}\left(z^{2}-1\right)\sin\left(\frac{t}{l}\right)}{2z^{2}}{\rm d}z\right)\wedge\left(n^{l}{\rm d}n^{k}-n^{k}{\rm d}n^{l}\right) \nonumber\\
	& & -\frac{l^{2}\left(z^{2}-1\right)^{2}\sin\left(\frac{t}{l}\right)\left(z^{2}+1\right)}{8z^{3}}{\rm d}n^{l}\wedge{\rm d}n^{k})\Bigl]\nonumber\\
	H_{5,i,j}  &= & \frac{1}{2}\epsilon_{ijkl}\Bigr[\left(\frac{l\left(z^{2}-1\right)^{2}\left(z^{2}+1\right)\sin\left(\frac{t}{l}\right)}{8z^{3}}{\rm d}t+\frac{l^{2}\left(z^{2}-1\right)\cos\left(\frac{t}{l}\right)}{2z^{2}}{\rm d}z\right)\wedge\left(n^{l}{\rm d}n^{k}-n^{k}{\rm d}n^{l}\right)\nonumber\\
	& & +  \frac{l^{2}\left(z^{2}-1\right)^{2}\left(z^{2}+1\right)\cos\left(\frac{t}{l}\right)}{8z^{3}}{\rm d}n^{l}\wedge{\rm d}n^{k}\Bigl]\nonumber\\
	H_{ijk}&=&\epsilon_{ijkl}dt\wedge\left(\frac{l\left(z^{2}+1\right)n^{l}}{2z^{2}}{\rm d}z+\frac{l\left(z^{6}+z^{4}-z^{2}-1\right)}{8z^{3}}{\rm d}n^{l}\right)\nonumber\\
	H_{0,m,5}&=&-\frac{l^{2}\left(z^{2}-1\right)^{3}}{16z^{3}}\left(\epsilon_{mijk}n^{k}{\rm d}n^{i}\wedge{\rm d}n^{j}\right)\nonumber
	\end{eqnarray}

\section{Analysis of the five-dimensional black hole with negative cosmological
constant}

\subsection{Energy as the mass charge}

Let us consider the solution of Einstein equations with negative cosmological
constant of the spherically symmetric black hole. The metric is equal
to
\begin{equation}
{\rm d}s^{2}=-\left(\frac{r^{2}}{l^{2}}+1-\frac{2m}{r^{2}}\right){\rm d}t^{2}+\left(\frac{r^{2}}{l^{2}}+1-\frac{2m}{r^{2}}\right)^{-1}{\rm d}r^{2}+r^{2}{\rm d}\Omega_{3}\,,\label{schwartzfielda}
\end{equation}
see e.g. equation (2.1) in \cite{key-3}. Here ${\rm d}\Omega_{3}$ denotes
the metric of the unit three-dimensional sphere. It turns out that if
we use the CYK tensor $C_{05}$ from the equation \ref{eq: c05} and
the Weyl tensor of the metric \ref{schwartzfielda}, we will end up with
three-form $T$ from the theorem \ref{thm:T zamkni=00003D00003D000119ta}
equal to
\begin{gather}
T=\frac{12m}{l}\omega\, ,
\end{gather}
where $\omega$ is the volume
three-form of the three-dimensional unit sphere. This result was calculated in Mathematica.
Therefore the quasilocal
charge equals to $\frac{24m\pi^{2}}{l}$. It means that mass of the
Anti-de Sitter is related with (asymptotic) CYK tensor $\frac{l}{24\pi^{2}}C_{05}=-\frac{l^{2}\left(z^{4}-1\right)}{96\pi^{2}z^{3}}{\rm d}t\wedge{\rm d}z$.
In this case the three-form \ref{eq:trojforma} does not depend on $z$
and is closed, so the energy charge in this case is not only asymptotic -- it is exact.
In the asymptotically flat case we have the so called  ADM mass defined as
\[
m_{ADM}:=\frac{1}{6\pi^{2}}\int_{S^{3}}\left(g_{ij,i}-g_{ii,j}\right){\dd}S^{j}\,.
\]
The coefficient $\frac{1}{6\pi^{2}}$ arises from the volume of three-dimensional sphere and from the coefficient in the Einstein equation in this dimension (see Appendix D in \cite{CJK2}).
More precisely,
\[2\gamma = \frac {2(n-1)\omega_{n-1}}{n-2} = \left\{ \begin{array}{c} 16\pi \quad \mbox{for} \; n=3\\ 6\pi^2 \quad \mbox{for} \; n=4 \end{array} \right. \, . \]
We think that in our case, which is not asymptotically flat, we should also multiply the result of integral on sphere by such  factor.
This means, that it is sufficient
to take $\frac{l}{4}C_{05}$ in the definition of the CYK~tensor
responsible for energy. In that case the asymptotic three-form will be
equal to the ADM form.

\subsection{Canonical coordinates on five-dimensional black hole with negative cosmological
constant}

Let us try to find the solution of the equation (\ref{eq:  eikona=00003D00003D00003D000142})
for the metric (\ref{schwartzfielda}). We have
\begin{gather}
1=\left\Vert \frac{l}{z}\dd z\right\Vert ^{2}=\left(\frac{l}{z}\right)^{2}\left(\frac{\partial z}{\partial r}\right)^{2}\left\Vert \dd r\right\Vert ^{2}=\frac{l^{2}}{z^{2}}\left(\frac{\partial z}{\partial r}\right)^{2}\left(\frac{r^{2}}{l^{2}}+1-\frac{2m}{r^{2}}\right)\kropka
\end{gather}
Because we expect that $z\sim\frac{1}{r}$, we demand $z>0$ and $\frac{\partial z}{\partial r}<0$.
For this reason we have
\begin{gather}
\frac{\partial z}{\partial r}=-\frac{z}{l}\frac{1}{\sqrt{\frac{r^{2}}{l^{2}}+1-\frac{2m}{r^{2}}}}\\
\log z+C=-\int\frac{\dd r}{l\sqrt{\frac{r^{2}}{l^{2}}+1-\frac{2m}{r^{2}}}}\kropka
\end{gather}
We substitute $w=\frac{l}{r}$ to obtain
\begin{gather}
-\int\frac{\dd r}{l\sqrt{\frac{r^{2}}{l^{2}}+1-\frac{2m}{r^{2}}}}=\int\frac{\dd w}{w^{2}\sqrt{w^{-2}+1-\frac{2mw^{2}}{l^{2}}}}=\int\frac{\dd w}{w\sqrt{1+w^{2}-w^{4}\frac{2m}{l^{2}}}}\kropka
\end{gather}
Denoting $b:=\frac{2m}{l^{2}}$ we get
\begin{gather}
z=\exp\left(\int\frac{\dd w}{w\sqrt{1+w^{2}-bw^{4}}}\right)\kropka\label{eikona=00003D00003D00003D000142 schwartzfield}
\end{gather}
It is easy to notice that if $z$ satisfies equation (\ref{eq:  eikona=00003D00003D00003D000142}),
then $\alpha z$ with $\alpha$ being arbitrary constant also satisfies
that equation. For this reason the constant arising from the integral
in equation (\ref{eikona=00003D00003D00003D000142 schwartzfield}) is
irrelevant. For small $w$ we have $\int\frac{\dd w}{w\sqrt{1+w^{2}-bw^{4}}}\simeq\log w$
so $z\simeq w$. We can also calculate the asymptotic
\begin{gather}
\int\frac{\dd w}{w\sqrt{1+w^{2}-bw^{4}}}\simeq\int\frac{\dd w}{w}\left(1-\frac{1}{2}\left(w^{2}-bw^{4}\right)+\frac{3}{8}\left(w^{2}-bw^{4}\right)^{2}\right)\nonumber \\
\simeq\int\frac{\dd w}{w}\left(1-\frac{1}{2}w^{2}+\left(\frac{1}{2}b+\frac{3}{8}\right)w^{4}\right)\simeq\log w-\frac{1}{4}w^{2}+\left(\frac{1}{8}b+\frac{3}{32}\right)w^{4}
\end{gather}
\begin{eqnarray}
z & = & w\exp\left(-\frac{1}{4}w^{2}+\left(\frac{1}{8}b+\frac{3}{32}\right)w^{4}\right)\nonumber \\
 & = & w\left(1-\frac{1}{4}w^{2}+\left(\frac{1}{8}b+\frac{3}{32}+\frac{1}{32}\right)w^{4}\right)\nonumber \\
 & = & w-\frac{1}{4}w^{3}+\frac{1}{8}\left(b+1\right)w^{5}+\cdots\,.
\end{eqnarray}
It is easy to check, that if $z=w+\alpha w^{3}+\beta w^{5}+\cdots$,
then $w=z-\alpha z^{3}+\left(3\alpha^{2}-\beta\right)z^{5}+\cdots$.
Therefore we have
\begin{gather}
w=z+\frac{1}{4}z^{3}+\frac{1-2b}{16}z^{5}\kropka
\end{gather}

We can now express the Schwarzschild metric in the coordinates $w$, $t$
and angles. We get the following metric
\begin{gather}
\dd s^{2}=\left(-1-w^{-2}+bw^{2}\right)\dd t^{2}+\left(\frac{l^{2}}{-bw^{6}+w^{4}+w^{2}}\right)\dd w^{2}+\left(\frac{l}{w}\right)^{2}{\rm d}\Omega\kropka
\end{gather}
From the construction we know that
\[
\left(\frac{l^{2}}{-bw^{6}+w^{4}+w^{2}}\right)\dd w^{2}=\left(\frac{l}{z}\right)^{2}\dd z^{2}\kropka
\]
For this reason we can now write everything in terms of $z$. We will
obtain the approximation of the real metric. Let's calculate the coefficient
that multiplies $\dd t^{2}$. Substituting $w=z+\omega z^{3}+\tau z^{5}$, we obtain
\begin{gather}
-1-w^{-2}+bw^{2}=-1-\left(z+\omega z^{3}+\tau z^{5}\right)^{-2}+b\left(z+\omega z^{3}+\tau z^{5}\right)^{2}\nonumber \\
= \frac{1}{z^{2}}\left[-z^{2}-\left(1+\omega z^{2}+\tau z^{4}\right)^{-2}+bz^{2}\left(z+\omega z^{3}+\tau z^{5}\right)^{2}\right]\nonumber \\
\simeq \frac{1}{z^{2}}\left[-z^{2}-1-1\left(-2\right)\left(\omega z^{2}+\tau z^{4}\right)-2\cdot3/2\left(\omega z^{2}\right)^{2}+bz^{4}\right]\nonumber \\
\simeq\frac{1}{z^{2}}\left[-1+z^{2}\left(-1+2\omega\right)+z^{4}\left(2\tau-3\omega^{2}+b\right)\right]
\nonumber \\
=\frac{1}{z^{2}}\left[-1-\frac{1}{2}z^{2}+\left(2\cdot\frac{1-2b}{16}-\frac{3}{16}+b\right)\right]\nonumber \\
= \frac{1}{z^{2}}\left[-1-\frac{1}{2}z^{2}+\left(-\frac{1}{16}+\frac{3}{4}b\right)z^{4}\right] \, .
\label{przy dt2 w schwartzfieldzie}
\end{gather}
Comparing this coefficient with analogous coefficient in the equation
(\ref{ads w sferycznych z n}) we see that we should get
$-\left(\frac{1+z^{2}}{2}\right)^{2}$.
However, looking at the equation (\ref{ads w sferycznych z n}) we see
that in the equation (\ref{przy dt2 w schwartzfieldzie}) coordinate
$z$ is $2$ times bigger, because it is behaving like $z\sim\frac{1}{r}$.
For this reason we are introducing $z=2\tilde{z}$. Now we have
\begin{gather}
-1-w^{-2}+bw^{2}=\frac{1}{\tilde{z}^{2}}\left[-\frac{1}{4}-\frac{1}{2}\tilde{z}^{2}+\left(-\frac{1}{4}+3b\right)\tilde{z}^{4}\right]\kropka
\end{gather}
This result is in accordance with equation (\ref{ads w sferycznych z n})
when $b=0$ (that is when $m=0$).

Now we only need to calculate the coefficient that multiplies $\dd\Omega$.
We have
\begin{gather}
\left(\frac{l}{w}\right)^{2}=\left(\frac{l}{z}\right)^{2}\left(1+\omega z^{2}+\tau z^{4}\right)^{-2}\nonumber \\
\simeq\left(\frac{l}{z}\right)^{2}\left(1-2\left(\omega z^{2}+\tau z^{4}\right)+3\left(\omega z^{2}\right)^{2}\right)\nonumber \\
=\left(\frac{l}{z}\right)^{2}\left(1-2\omega z^{2}+z^{4}\left(-2\tau+3\omega^{2}\right)\right)\nonumber \\
=\left(\frac{l}{z}\right)^{2}
\left(1-\frac{1}{2}z^{2}+z^{4}\left(-2\frac{1-2b}{16}+\frac{3}{16}\right)\right)\nonumber \\
=\left(\frac{l}{z}\right)^{2}\left(1-\frac{1}{2}z^{2}+\left(\frac{1}{16}+\frac{1}{4}b\right)z^{4}\right)
\, .
\end{gather}
Again let's express it with function $\tilde{z}$, so we get
\begin{gather}
\left(\frac{l}{w}\right)^{2}=\left(\frac{l}{\tilde{z}}\right)^{2}\left[\frac{1}{4}-\frac{1}{2}\tilde{z}^{2}+\left(\frac{1}{4}+b\right)\tilde{z}^{4}\right]\kropka
\end{gather}
This result is in accordance with equation (\ref{ads w sferycznych z n}).
We finally obtain
\begin{gather}
{\rm d}s^{2}\simeq\left(\frac{l}{\tilde{z}}\right)^{2}\left[\left(3b\tilde{z}^{4}-\left(\frac{1+\tilde{z}^{2}}{2}\right)^{2}\right)\left(\frac{\dd t}{l}\right)^{2}+\dd\tilde{z}^{2}+\left(\left(\frac{1-\tilde{z}^{2}}{2}\right)^{2}+b\tilde{z}^{4}\right)\dd\Omega\right]\,.
\end{gather}
This metric is in accordance with the equation (\ref{ads w sferycznych z n}).
It turns out, that our metric differs from Anti-de Sitter metric by
terms that have rank 4 in $\tilde{z}$. We now use the fact that $b=\frac{2m}{l^{2}}$
to express the metric in terms of $m$. We get
\begin{equation}
{\rm d}s^{2}\simeq \left(\frac{l}{\tilde{z}}\right)^{2}
\left[\left(\frac{6m}{l{^{2}}}\tilde{z}^{4}-\left(\frac{1+\tilde{z}^{2}}{2}\right)^{2}\right)\left(\frac{\dd t}{l}\right)^{2}+\dd\tilde{z}^{2}
+\left(\left(\frac{1-\tilde{z}^{2}}{2}\right)^{2}
+\frac{2m}{l{^{2}}}\tilde{z}^{4}\right)\dd\Omega\right] \, . \label{dziura asymptotyka}
\end{equation}
The above asymptotic form is in accordance with the general form of the asymptotically
Anti-de Sitter metrics given by (1.2) in 
\cite{key-2} which simply means that Schwarzschild-AdS is an asymptotically
Anti-de Sitter space-time.

\section{Conclusion}
We propose a new construction of CYK tensors in AdS5 using the observation that {\em constant tensors} in the ambient space restricted to the pseudosphere AdS5 generate all solutions of CYK equation.
We would like to stress that Theorems in Section 3 are nice tools and we show in Appendix how to use them to construct in explicit form standard 4D conformal covector fields. One can argue that AdS5 is conformally equivalent to 5D Minkowski, hence, using conformal transformation, we can translate solution in flat space to the solution in constant curvature space. However, the construction of solutions in flat space and corresponding conformal rescaling is not so simple. We think that our construction is simple and natural, one can say that the CYK tensors in AdS are simpler than in flat Minkowski because they are naturally obtained from constant tensors. Obviously, CYK tensors in Minkowski can be reconstructed from AdS via conformal transformation or by limiting procedure (in the tangent space).

It turns out that in the case of five-dimensional Anti-de Sitter space-time one
can carry out constructions very similar to those in the four-dimensional
case. Specifically in the five-dimensional case one can find all conformal
Yano-Killing tensors in a way that is analogous to the reasoning in 
\cite{chas} which solves the same problem in the four-dimensional case.

CYK tensors obtained in {\rm AdS5} enable us
to define quasi-local charges that have good asymptotic properties.
We have chosen the CYK tensor which defines the energy for the example
of five-dimensional Schwarzschild blackhole. Probably in the case of
the five-dimensional Kerr blackhole with negative cosmological constant
it is possible to find CYK tensor which is responsible for the angular
momentum, but it seems to be quite heavy calculation, so it won't
be analyzed in this publication.

\vspace{0.5 cm}

	\noindent {\sc Acknowledgements}
This work was supported in part by Narodowe Centrum Nauki (Poland) under Grant No. 2016/21/B/ST1/00940.

\appendix

\section{Conformal Killing one-forms in four dimensions}

We will now find all conformal Killing one-forms on four-dimensional
Anti-de Sitter space-time. According to equation (1.8) in \cite{stab} Anti-de Sitter metric has the following form
\begin{gather}
\tilde g=\frac{l^2}{\cos^2x}\left(-{\rm d}t^{2}+{\rm d}x^{2}+\sin^{2}\left(x\right)\sigma\right)\,. \label{prawdziwa}
\end{gather}
where $\sigma$ is a metric on a two-dimensional sphere, $l\in \mathbb{R}$  is a size of our Anti-de Sitter spacetime, $t\in\mathbb{R},\,x\in\left[0,\pi\right]$.
$t$ has index $0$, coordinates on the sphere have indices $1$ and
$2$, and $x$ has index $3$.
According to the theorem \ref{confkilloneconf} we can find conformal Killing one-forms on the conformally equivalent metric
\begin{gather}
g=-{\rm d}t^{2}+{\rm d}x^{2}+\sin^{2}\left(x\right)\sigma\, , \label{konfrow}
\end{gather}
and later multiply found one-forms by the conformal factor $\frac{l^2}{\cos^2x}$. This our strategy.

We are now abandoning previous index conventions. Let's denote that indices $A,B$ go from $1$ to $2$. They are used
to parameterize the sphere. We will also use $;$ to denote covariant derivatives on the spheres of constant $x$ and $t$ (using metric connection of this sphere). Greek indices go from $0$ to $3$. Additionally $\eta$ is a metric induced
on the slice of constant time. We have therefore $\eta_{AB}=\sin^{2}x\cdot\sigma_{AB}$,
$\eta_{3A}=0$ and $\eta_{33}=1$. Latin indices denote spacial coordinates. In the following (unless stated otherwise) we will use $|$ to denote covariant derivative on the spacial slice (slice with constant $t$).
We see that according to the theorem \ref{cof konforemnej jest konforemne}
if we pull back conformal Killing one-form to the slice of constant
time, we will obtain conformal Killing one-form on the slice. This
is so because the slice has external curvature equal to 0. Analogously
if we pull back conformal Killing one-form to the slice of constant
time and constant $x$, we will obtain conformal Killing one-form on
this slice.

We need to calculate Christoffel symbols  $\Gamma^{\alpha}{_{\beta\gamma}}$ for the metric $g$.
Coefficients $\Gamma^{A}{_{BC}}$ are the same as Christoffel symbols
on the unit sphere. $\Gamma^{A}{_{3B}}=\frac{1}{2}\eta^{AC}\eta_{CA|3}=\cot x\delta_{B}^{A}$,
$\Gamma^{3}{_{AB}}=-\frac{1}{2}g^{33}\eta_{AB|3}=-\cot x\eta_{AB}$.\\
 Those are all non vanishing Christoffel symbols $\Gamma{^{\alpha}}{_{\beta\gamma}}$.

\subsection{Two dimensional problem on the sphere}

$\xi$ is our conformal Killing one-form on the whole space-time (with metric $g$ conformally equivalent to the metric of Anti-de Sitter).
Let's now consider its pullback to the slice of constant $t$ and
$x$ (to the sphere). As we previously stated, this restriction is
a conformal Killing one-form on the sphere, so it satisfies the equation
\begin{gather}
\xi_{A;B}+\xi_{B;A}-\eta_{AB}\xi^{C}{_{;C}}=0\,,\label{eq:trzecie}
\end{gather}
where $;$ denotes covariant derivative on the sphere (for a given
point we pull back one-form to the sphere and then we covariantly
differentiate it using metric connection on the sphere). Equality
(\ref{eq:trzecie}) means, that $\xi_{A}$ is a dipole one-form. This
can be shown in the following way. 
For every (pseudo)-Riemannian manifold $K$ with Riemann tensor $R$ and one-form $\omega$ we have
\begin{gather}
\omega_{c|ba}-\omega_{c|ab}=R_{abc}{^{d}}\omega_{d}\, . \label{Riemana}
\end{gather}
We are using here $|$ to denote Levi-Civita derivative on the manifold $K$.

This can be generalized to the following identity
\begin{gather}
2T^{\alpha_{1}\alpha_{2}\cdots\alpha_{n}}{_{\beta_{1}\cdots\beta_{m}|\left[ba\right]}}=\sum_{k=1}^{m}R_{ab\beta_{k}}{^{d}}T^{\alpha_{1}\cdots\alpha_{n}}{_{\beta_{1}\cdots d\cdots\beta_{m}}}-\sum_{k=1}^{n}R_{abd}{^{\alpha_{k}}}T^{\alpha_{1}\cdots d\cdots\alpha_{n}}{_{\beta_{1}\cdots\beta_{m}}}\,.
\end{gather}
In the future calculations we will use the character $|$ to denote
the covariant derivative on the spacial part of Anti-de Sitter space-time with respect to the metric  induced from  the metric $g$ which is conformally equivalent to the metric of Anti-de Sitter.

We are choosing a sphere of constant $x$ and  $t$. This sphere
has natural structure of pseudo-Riemannian manifold which
is the same as the structure of sphere, which has radius $r=\sin x$
(that means that there exists isomorphism from the category of pseudo-Riemannian
manifolds between a sphere of constant $x$ and $t$ and a sphere of
the radius $r=\sin x$). We know that Ricci tensor on the sphere
of the radius $r$ equals to $R_{AB}=\frac{1}{r^{2}}\eta_{AB}$. For
this reason Riemann tensor equals $R_{ABCD}=\frac{1}{r^{2}}\left(\eta_{AC}\eta_{BD}-\eta_{AD}\eta_{BC}\right)$.
In the following calculations for the sphere we will lower and rise
the indices using metric $\eta$ induced from the metric $g$ from the equation \ref{konfrow}. $\epsilon_{AB}$
will denote the volume form of our sphere with respect to the metric
$\eta$. According to the Hodge--Kodaira theorem each one-form can
be expressed as the sum of external derivative, coderivative and harmonic
form. It is also well known that there are no harmonic one-forms on the
sphere. For this reason we have
\begin{gather}
\xi_{A}=\overset{1}{v}_{;A}+\epsilon_{A}{^{B}}\overset{2}{v}_{;B}\,,\label{eq:hodgea kodairy}
\end{gather}
where $\overset{1}{v}$ are $\overset{2}{v}$ are some functions.

We will show that $\xi{_{A}}$ is a dipole function. We can find
functions $\overset{1}{v}$ and $\overset{2}{v}$ in the following
way
\begin{equation}
\xi_{A}{^{;A}}=\overset{1}{v}_{;A}{^{A}}\,,\label{eq:  div pom pom}
\end{equation}
(we are using convention that all indices to the right of the $;$
are differentiating the tensor field) and
\begin{gather}
\xi_{A;C}\epsilon^{AC}=\epsilon_{A}{^{B}}\overset{2}{v}_{;BC}\epsilon^{AC}=\overset{2}{v}_{;A}{^{A}}\,.\label{eq: rot pom pom}
\end{gather}
Therefore if we will show that $\xi_{A;C}\epsilon^{AC}$ and $\xi_{A}{^{;A}}$
are dipole functions, then also $\xi_{A}$ will be a dipole one-form.

Let us differentiate equation \ref{eq:trzecie} covariantly with index
$A$ at the top. We will end up with
\begin{eqnarray}
0& =& \xi_{A;B}{^{A}}+\xi_{B;A}{^{A}}-\xi^{C}{_{;CB}}\nonumber \\
 & = & \xi_{A}{^{;A}}{_{B}}+R^{A}{_{BA}}{^{D}}\xi_{D}+\xi_{B;A}{^{A}}-\xi^{C}{_{;CB}}\nonumber \\
 & = & \frac{1}{r^{2}}\xi_{B}+\xi_{B;A}{^{A}}=0\,.\label{eq:trzecie zrozniczkowane zwezone}
\end{eqnarray}
Equality (\ref{eq:trzecie zrozniczkowane zwezone}) can be written in
the form
\begin{gather}
\left(\overset{0}{\Delta}+1\right)\xi=0\,,\label{zedipol}
\end{gather}
where $\overset{0}{\Delta}$ denotes Laplacian created from the structure
of the metric of the unit sphere, and $1$ denotes identity operator.
This means that we are pulling back $\xi$ to the sphere, and then
we are  using the metric of the unit sphere to calculate Laplacian
of resulting $\xi_{A}$. Equality (\ref{zedipol}) means that $\xi_{A}$
is a dipole one-form. Now let's covariantly differentiate the equation
(\ref{eq:trzecie zrozniczkowane zwezone}) and then contract resulting
index with $B$.
\begin{eqnarray}
0 &= &\frac{1}{r^{2}}\xi_{B}{^{;B}}+ \xi_{B;A}{^{BA}}+R^{BA}{_{B}}{^{D}}\xi_{D;A}+R^{BA}{_{A}}{^{D}}\xi_{B;D} \nonumber \\
&= &\frac{1}{r^{2}}\xi_{B}{^{;B}}+\xi_{B;A}{^{BA}} \nonumber \\
&= & \frac{1}{r^{2}}\xi_{B}{^{;B}}+\xi_{B}{^{;B}}{_{A}}{^{A}}+ \left(R^{B}{_{AB}}{^{D}}\xi_{D}\right)^{;A}=\xi_{B}{^{;B}}{_{A}}{^{A}}+\frac{2}{r^{2}}\xi_{A}{^{;A}}\label{eq:  nie rotacyjna pom}
\end{eqnarray}
This proves that $\xi_{A}{^{;A}}$ is a dipole function.

To prove that $\xi_{B;C}\epsilon^{BC}$ is a dipole
function we start with the following identity
\begin{gather}
R_{CAB}{^{D}}\epsilon^{BC}=-\frac{1}{r^{2}}\delta_{C}{^{D}}g_{AB}\epsilon^{BC}=-\frac{1}{r^{2}}\epsilon_{A}{^{D}}\,.\label{eq:pom}
\end{gather}
We will now differentiate covariantly the equality (\ref{eq:trzecie zrozniczkowane zwezone})
and then contract the result with $\epsilon^{BC}$. We will end up
with
\begin{eqnarray}
0 & = & \frac{1}{r^{2}}\xi_{B;C}\epsilon^{BC}+\xi_{B;A}{^{A}}{_{C}}\epsilon^{BC}\nonumber \\
 & = & \frac{1}{r^{2}}\xi_{B;C}\epsilon^{BC}+\xi_{B;AC}{^{A}}\epsilon^{BC}+R_{C}{^{A}}{_{B}}{^{D}}\xi_{D;A}\epsilon^{BC}+R_{C}{^{A}}{_{A}}{^{D}}\xi_{B;D}\epsilon^{BC}\nonumber \\
 & = & \frac{1}{r^{2}}\xi_{B;C}\epsilon^{BC}+\xi_{B;AC}{^{A}}\epsilon^{BC}-\frac{1}{r^{2}}\epsilon^{AD}\xi_{D;A}-\frac{1}{r^{2}}\xi_{B;D}\epsilon^{BD}\nonumber \\
 & = & \frac{1}{r^{2}}\xi_{B;C}\epsilon^{BC}+\xi_{B;CA}{^{A}}\epsilon^{BC}+R_{CAB}{^{D}}\xi_{D}{^{;A}}\epsilon^{BC}\nonumber \\
 & = & \frac{1}{r^{2}}\xi_{B;C}\epsilon^{BC}+\left(\xi_{B;C}\epsilon^{BC}\right)_{;A}{^{A}}-\frac{1}{r^{2}}\epsilon_{A}{^{D}}\xi_{D}{^{;A}}\nonumber \\
 & = & \frac{2}{r^{2}}\xi_{B;C}\epsilon^{BC}+\left(\xi_{B;C}\epsilon^{BC}\right)_{;A}{^{A}}\,.\label{eq: rotacyjna pom}
\end{eqnarray}
The last equality means that $\xi_{B;C}\epsilon^{BC}$ is a dipole
function.

Equations (\ref{eq: rotacyjna pom}), (\ref{eq:  nie rotacyjna pom}),
(\ref{eq:  div pom pom}) and (\ref{eq: rot pom pom}) prove that $\Delta\Delta\overset{1}{v}$
and $\Delta\Delta\overset{2}{v}$ are dipole functions. Dipole functions
belong to the eigenspace of Laplacian with non zero eigenvalue. That is why
Laplacian $\Delta$ acts on them as an isomorphism. For this reason
functions  $\overset{1}{v}$ and $\overset{2}{v}$ are sums of dipole
functions and elements of the kernel of $\Delta$, which are monopole
functions. According to the equation (\ref{eq:hodgea kodairy}) monopole
parts of these functions don't matter because in the equation (\ref{eq:hodgea kodairy})
functions $\overset{1}{v}$ and $\overset{2}{v}$ are differentiated.

\subsection{Expanding to the spacial slice}

In this section we will denote covariant derivative with respect to the slice of constant time with the character $|$. We have to remember here that we are using the metric \ref{konfrow}. The conformal Killing one-forms on the sphere enable one
to find conformal Killing one-forms on the whole slice of the constant
$t$. We are calculating the covariant derivatives
\[
\xi_{k|l}=\xi_{k,l}-\Gamma^{i}{_{kl}}\xi_{i} \, ,
\]
\[
\xi_{3|3}=\xi_{3,3} \, ,
\]
\[
\xi_{3|A}=\xi_{3,A}-\Gamma^{i}{_{3A}}\xi_{i}=\xi_{3,A}-\cot x\,\xi_{A} \, ,
\]
\[
\xi_{A|3}=\xi_{A,3}-\cot x\,\xi_{A} \, ,
\]
\[
\xi_{A|B}=\xi_{A;B}-\Gamma^{3}{_{AB}}\xi_{3}=\xi_{A;B}+\cot x\,\eta_{AB}\xi_{3} \, .
\]
On the spacial slice we have
\begin{gather}
\xi_{k|l}+\xi_{l|k}=\alpha\cdot\eta_{kl}\,.
\end{gather}
Let us calculate spacial derivative
\begin{gather}
\xi^{k}{_{|k}}=\eta^{33}\xi_{3|3}+\eta^{AB}\xi_{A|B}=
\xi_{3,3}+\xi^{A}{_{;A}}+2\cot x\cdot\xi_{3} \, , \label{eq:trojdywergencja}
\end{gather}
We have $2\xi^{k}{_{|k}}=3\alpha$ so
\begin{equation}
\alpha=\frac{2}{3}\xi^{k}{_{|k}}\,.\label{alpha na podstawie ciecia stalegoczasu}
\end{equation}
It follows from the conformal Killing equation that
\begin{equation}
\xi_{3|A}+\xi_{A|3}=0\, ,\label{eq:pom3}
\end{equation}
\begin{equation}
2\xi_{3|3}=\alpha\, ,\label{pom}
\end{equation}
\[
\xi_{A|B}+\xi_{B|A}=\alpha\eta_{AB}\, ,
\]
so contracting the last equation with $\eta^{AB}$ we get
\begin{equation}
\eta^{AB}\xi_{A|B}=\alpha\,.\label{eq:pom2}
\end{equation}
We derive equation (\ref{eq:pierwsze}) from the equation (\ref{eq:pom3}),
whereas combined equations (\ref{pom}) and (\ref{eq:pom2}) lead to $\xi^{A}{_{;A}}+2\cot x\cdot\xi_{3}=\eta^{AB}\xi_{A|B}=\alpha=2\xi{_{3,3}}$.
This last equation is equivalent to the equation (\ref{eq:drugie}).\\[2ex]
\fbox{\begin{minipage}[c]{13cm}%
\begin{gather}
\xi_{3,A}+\xi_{A,3}-2\cot x\cdot\xi_{A}=0\label{eq:pierwsze}
\end{gather}
\begin{gather}
2\xi_{3,3}=\xi^{A}{_{;A}}+2\cot x\cdot\xi_{3}\label{eq:drugie}
\end{gather}
\phantom{XX} %
\end{minipage}}\\[2ex]

We apply covariant derivative $;A$ to the equation
(\ref{eq:pierwsze}) and  we obtain 
\begin{gather}
\xi_{3;A}{^{A}}+\xi_{A;B,3}\eta^{AB}-2\cot x\cdot\xi_{A}{^{;A}}=0\,.\label{eq:o x itrzy}
\end{gather}
Here we used the fact that partial derivative $\partial_{3}$ and
covariant derivative $;A$ commute. This is the consequence of the
fact, that Christoffel symbols do not depend on $x$. Equation (\ref{eq:o x itrzy})
proves that $\Delta\xi_{3}$ is a dipole function, because the rest
of this equality is a dipole function. For this reason $\xi_{3}$
has only monopole and dipole parts.

We can also rewrite the equation (\ref{eq:drugie}) in the form
\begin{gather}
2\left(\frac{\xi_{3}}{\sin x}\right)_{,3}\sin x=\xi^{A}{_{;A}}\,.\label{eq:pierwszeprzeksztalcone}
\end{gather}
We see therefore that monopole and dipole parts of $\xi\subbb 3$
evolve independently. Let's rewrite the equality (\ref{eq:pierwszeprzeksztalcone})
in the form
\begin{gather}
2\left(\frac{{^{m}}\xi_{3}}{\sin x}\right)_{,3}\sin x=0\\
2\left(\frac{{^{d}}\xi_{3}}{\sin x}\right)_{,3}\sin x=\xi^{A}{_{;A}}\,,\label{eq:xi dipolowaa}
\end{gather}
where 
${^{m}}\xi_{3}$ is the monopole part of $\xi{_{3}}$, whereas 
${^{d}}\xi_{3}$ is the dipole part of $\xi{_{3}}$. From the first
of those equations we get
\begin{gather}
{^{m}}\xi_{3}=a\sin x\,,
\end{gather}
where $a$ is some constant. It is denoted with small letter because
it is a monopole function. Functions that are dipole will be denoted
with capital letters.

We can also rewrite the equation (\ref{eq:pierwsze}) in the form
\begin{gather}
\xi_{3}{^{;A}}+\xi^{A}{_{,3}}=0\,.\label{eq:dziwne}
\end{gather}
Equation \ref{eq:dziwne} was obtained by noticing that equation (\ref{eq:dziwne})
can be expressed as $\xi^{A}{_{,3}}=\left(\sin^{-2}x\,\sigma^{AB}\xi_{B}\right)_{,3}=\sin^{-2}x\,\sigma^{AB}\xi_{B}{_{,3}}-2\sin^{-3}x\cos x\cdot\sigma^{AB}\xi_{B}$
(here $\sigma{^{AB}}$ is a metric inverse to $\sigma{_{AB}}$) so
it looks like $\xi_{B,3}-2\cot x\xi_{B}$ with raised index. It is
worth to remember that here $\xi_{3}$ is treated as a function, so
in the expression $\xi_{3}{^{;A}}$ covariant derivative acts as partial
derivative. We remember that in the covariant derivative $;A$ Christoffel
symbols $\Gamma^{A}{_{BC}}$ are independent from $x$. This means
that covariant derivative $;A$ commutes with $\partial_{3}$. For
this reason we can calculate covariant derivative $;A$ of the equation
\ref{eq:dziwne} and contract the indices. We end up with
\begin{gather}
\xi_{3}{^{;A}}{_{A}}+\left(\xi^{A}{_{;A}}\right)_{,3}=0\,.\label{eq:nowedziwne}
\end{gather}
This equation can be rewritten
as follows
\begin{gather}
\frac{1}{\sin^{2}x}\overset{0}{\triangle}\xi_{3}+\left(\xi^{A}{_{;A}}\right)_{,3}=0\,,
\end{gather}
where $\overset{0}{\triangle}$ denotes Laplacian on the unit sphere.
We therefore see that
\begin{gather}
\left(\xi^{A}{_{;A}}\right)_{,3}=2\frac{\suppp d\xi_{3}}{\sin^{2}x}\,.\label{eq:zmanadywergencji}
\end{gather}
By combining equations (\ref{eq:zmanadywergencji}) and (\ref{eq:xi dipolowaa})
we get
\begin{gather}
\left(\left(\frac{\suppp d\xi_{3}}{\sin x}\right)_{,3}\sin x\right)_{,3}=\frac{\suppp d\xi_{3}}{\sin^{2}x}\,.\label{eq:naksi}
\end{gather}
This equation can be solved. We introduce $u=\frac{\suppp d\xi_{3}}{\sin x}$
and $z=\int\frac{dx}{\sin x}$. Now the equation (\ref{eq:naksi}) has
the following form $\frac{d^{2}u}{dz^{2}}=u$. It has solution $u=Be^{z}+Ce^{-z}$,
where $B,C$ are independent from $z$. This can be written as $u=\log\left(\sin\left(\frac{x}{2}\right)\right)-\log\left(\cos\left(\frac{x}{2}\right)\right)$.
\begin{gather}
u=B\frac{\sin\left(\frac{x}{2}\right)}{\cos\left(\frac{x}{2}\right)}+C\frac{\cos\left(\frac{x}{2}\right)}{\sin\left(\frac{x}{2}\right)}
\end{gather}
\begin{gather}
\suppp d\xi_{3}=\sin\left(x\right)\left(B\tan\frac{x}{2}+C\cot\frac{x}{2}\right)
\end{gather}
In the last equation left hand side is a dipole function, so $A,B$
are also  dipole functions.
\begin{eqnarray}
\xi_{3} &= & \suppp m\xi_{3}+\suppp d\xi_{3}\nonumber \\ &= &
a\sin x+\sin\left(x\right)\left(B\tan\frac{x}{2}+C\cot\frac{x}{2}\right)=a\sin x+2\left(B\sin^{2}\frac{x}{2}+C\cos^{2}\frac{x}{2}\right) \nonumber \\ &= &
a\sin x+2\left(B\sin^{2}\frac{x}{2}+C\cos^{2}\frac{x}{2}\right) \nonumber \\ &= &
a\sin x+B\left(1-\cos x\right)+C\left(1+\cos x\right) \nonumber \\ &= &
a\sin x+K-J\cos x
\end{eqnarray}
Constants $K$ and $J$ are replacing the constants $B$ and $C$
in the following way: $K=B+C$ and $J=B-C$.

We now use the equation (\ref{eq:drugie}) to get
\begin{gather}
\xi^{A}{_{;A}}=2\xi_{3,3}-2\cot x\cdot\xi_{3} \, , \nonumber \\
2a\cos x+2J\sin x-2\cot x\left(a\sin x+K-J\cos x\right)\nonumber 
 = -2K\cot x+2J\frac{1}{\sin x} \, .
\end{gather}
Now we only need to find the rotational part of $\xi_{A}$. In the
previous subsection we defined $\epsilon$ as the volume form of the
sphere of constant $x$ and $t$. Now we want to think about the $\epsilon$
as a tensor on the whole Anti-de Sitter space-time with conformally equivalent metric from the equation \ref{konfrow}. We define $\epsilon$ on
the whole Anti-de Sitter space-time by imposing relations $\epsilon_{\mu3}=\epsilon_{\mu0}=0$,
$\epsilon_{\mu\nu}=-\epsilon_{\nu\mu}$, and $\epsilon_{AB}$ is the
metric volume form of the sphere of constant $x$ and $t$ with respect
to the metric induced $g$ from the equation \ref{konfrow}. It is easy to calculate
that $\epsilon{^{AB}}_{,3}=-2\textrm{ctg}x\epsilon^{AB}$. We are
differentiating covariantly the equation (\ref{eq:pierwsze}) with respect
to the $;B$ index, and then we are contracting resulting equation
with $\epsilon^{AB}$. By taking into account the commutation of the
$\partial_{3}$ and $;A$ we get
\begin{gather}
\xi_{A;B,3}\epsilon^{AB}-2\cot x\cdot\xi_{A;B}\epsilon^{AB}=0\, .
\end{gather}
We have therefore
\begin{gather}
\xi_{A;B,3}\epsilon^{AB}= \left(\xi_{A;B}\epsilon^{AB}\right)_{,3}-\xi_{A;B}\epsilon^{AB}{_{,3}}
\end{gather}
and finally
\begin{gather}
\xi_{A;B}\epsilon^{AB}=D \, ,
\end{gather}
where $D$ is a dipole function.

In that way we obtained the following
\begin{gather}
\xi_{A;B}\epsilon^{AB}=D\\
\xi_{A}{^{;A}}=-2K\cot x+2J\frac{1}{\sin x}\\ \label{pomdivergencja}
\xi_{3}=a\sin x+K-J\cos x
\end{gather}
The space of solutions is ten-dimensional, which is the maximal possible
number in three dimensions.

We now use the equation (\ref{eq:trojdywergencja}) to get
\begin{eqnarray}
\xi_{k}{^{|k}} &= & \xi_{3,3}+\xi^{A}{_{;A}}+2\cot x\cdot\xi_{3}\nonumber \\
%
 &=& 3a\cos x+3J \sin x  \,.
\end{eqnarray}
This means, that $\alpha$ from the equation $\xi_{\alpha|\beta}+\xi_{\beta|\alpha}=\alpha g_{\alpha\beta}$
(Greek indices may be both spacial and temporal) is equal to
\begin{gather}
\alpha=\frac{2}{3}\xi_{k}{^{|k}}=2a\cos x+2J\sin x\,. \label{pomalpha}
\end{gather}

\subsection{One-forms in space-time}

We should now consider the dependence of $a,J,K,D$, which are functions that characterize $\xi$, on time.\\[1ex]
\noindent %
\fbox{\begin{minipage}[c]{13cm}
\begin{gather}
2\xi_{0,0}=-\alpha\label{eq:t pierwsze}\\
\xi_{0,3}+\xi_{3,0}=0\label{t  drugie}\\
\xi_{0,A}+\xi_{A,0}=0\label{t trzecie}\\
\nonumber
\end{gather}
\end{minipage}}
\\[1ex]

\noindent Let us apply the covariant derivative ${;A}$ to the equation
(\ref{t trzecie}). We obtain the following equation.
\begin{gather}
\xi_{0;A}{^{A}}=-\xi{^{A}}_{;A}{_{,0}}. \label{eq: xi 0 pierwsze}
\end{gather}
This equation proves that $\overset{0}{\triangle}\xi_{0}$
is a dipole function, because $\xi{^{A}}_{;A}$ is a dipole function.

For this reason
\begin{gather}
\xi_{0}=\suppp d\xi_{0}+\suppp m\xi_{0}\,.
\end{gather}
We will now differentiate the equation (\ref{eq: xi 0 pierwsze}) with
respect to time. We obtain
\begin{gather}
\xi_{0,0;A}{^{A}}=-\xi_{A;}{^{A}}{_{,00}}\,.
\end{gather}
We can use the last equation together with equation  \ref{eq:t pierwsze} and \ref{pomalpha} to obtain
\begin{gather}
\xi_{A;}{^{A}}{_{,00}}=\frac{1}{2}\alpha_{;A}{^{A}}=\frac{1}{2}\frac{1}{\sin^{2}x}\left(-2\right)2J\sin x
\end{gather}
We now use the equation \ref{pomdivergencja} to obtain 
\begin{gather}
-2K_{,00}\cot x+2J_{,00}\frac{1}{\sin x}=-2J\frac{1}{\sin x}\,.\label{eq:  pierwsze na BC}
\end{gather}
To obtain the second equation for the coefficients $B,C$ we have
to differentiate the equation (\ref{t  drugie}) with respect to time.

\[
\xi_{3,00}=-\xi_{0,03}=\frac{1}{2}\alpha_{,3}=\frac{1}{2}\left(-2a\sin x+2J\cos x\right)
\]
so
\[
a_{,00}\sin x+K_{,00}-J_{,00}\cos x=-a\sin x+J\cos x\,.
\]
This equation splits into the dipole and monopole parts
and we get
\begin{gather}
a_{,00}=-a \, , \label{eq: pierwsze liniowe}\\
K_{,00}-J_{,00}\cos x=J\cos x \, , \label{drugie liniowe}\\
J_{,00}-K_{,00}\cos x=-J\, . \label{trzecie liniowe}
\end{gather}
The last equation is equivalent to the equation (\ref{eq:  pierwsze na BC}).
Equation (\ref{eq: pierwsze liniowe}) has a solution
\begin{gather}
a=a_{0}\sin t+a_{1}\cos t \, .
\end{gather}
Let us multiply the equation (\ref{drugie liniowe}) by $\cos x$
and add the result to the equation (\ref{trzecie liniowe}). We get
\begin{gather}
J_{,00}\sin^{2}x=-J\sin^{2}x \, .
\end{gather}
We have therefore
\begin{gather}
J_{,00}=-J
\end{gather}
and
\begin{gather}
J=J_{0}\sin t+J_{1}\cos t \, .\label{eq: jot}
\end{gather}
Let us add the equation (\ref{drugie liniowe}) to the equation (\ref{trzecie liniowe})
multiplied by $\cos x$. We obtain
\begin{gather}
K_{,00}=0 \, , \\
K=G+Ht \, .
\end{gather}
We can also covariantly differentiate the equation (\ref{t trzecie})
with respect to the index $;B$ and contract the result with $\epsilon^{AB}$. We obtain
\begin{gather}
\left(\xi_{A;B}\epsilon^{AB}\right)_{,0}=0\label{eq:rotacja po czasie}
\end{gather}
Equation (\ref{eq:rotacja po czasie}) has a solution
\begin{gather}
\xi_{A;B}\epsilon^{AB}=D \, , \label{eq: rot xi}
\end{gather}
where $D$ is a dipole function which is time independent. Additionally
\begin{eqnarray}
\xi_{3} &= & a\sin x+K-J\cos x \nonumber \\
 &= & \left(a_{0}\sin t+a_{1}\cos t\right)\sin x+G+Ht-\left(J_{0}\sin t+J_{1}\cos t\right)\cos x
\end{eqnarray}
and
\begin{gather}
\xi_{A}{^{;A}}=-2K\cot x+2J\frac{1}{\sin x} \, . 
\end{gather}
Now we only have to calculate the coefficient  $\xi_{0}$. Let's calculate
\begin{gather}
\alpha=2a\cos x+2J\sin x \, .
\end{gather}
We can use the equation (\ref{eq:t pierwsze}) to get
\begin{gather}
\xi_{0,0}=-\frac{1}{2}\alpha=-\left(a_{0}\sin t+a_{1}\cos t\right)\cos x+\left(-J_{0}\sin t-J_{1}\cos t\right)\sin x \, .
\end{gather}
We have therefore
\begin{gather}
\xi_{0}=\left(a_{0}\cos t-a_{1}\sin t\right)\cos x+\left(J_{0}\cos t-J_{1}\sin t\right)\sin x+F\,,\label{eq: xi 0}
\end{gather}
where $F$ is a certain function with both monopole and dipole parts,
which are time-independent. We don't know yet how they depend on
$x$. If we now substitute our results to (\ref{t  drugie}) we will
find that $H=0$. More precisely,
\begin{gather}
\xi_{3,0}=\left(a_{0}\cos t-a_{1}\sin t\right)\sin x+H-\left(J_{0}\cos t-J_{1}\sin t\right)\cos x\\
\xi_{0,3}=-\left(a_{0}\cos t-a_{1}\sin t\right)\sin x+\left(J_{0}\cos t-J_{1}\sin t\right)\cos x+F_{,3}\,.
\end{gather}
We have therefore from the equation (\ref{t  drugie})
\begin{gather}
H+F_{,3}=0\,.\label{eq: pierwsze na M}
\end{gather}
Equation (\ref{eq: xi 0 pierwsze}) proves that
\begin{gather}
\xi_{0;A}{^{;A}}=\frac{-2}{\sin^{2}x}\left(J_{0}\cos t-J_{1}\sin t\right)\sin x+F_{;A}{^{A}} \, , \\
\xi_{A}{^{;A}}{_{,0}}=-2H\cot x+\frac{2}{\sin x}\left(J_{0}\cos t-J_{1}\sin t\right) \, ,
\end{gather}
so
\begin{gather}
0=F_{;A}{^{A}}-2H\cot x=\frac{-2}{\sin^{2}x}{\suppp dF} -2H\cot x\,.
\end{gather}
\begin{gather}
{\suppp dF}=-H\sin x\cos x \, . \label{M dipo}
\end{gather}
This combined with the equation (\ref{eq: pierwsze na M}) leads to
$H=0$. $F$ is independent from $x$, $t$ and angles. Let's denote
this constant quantity as $F=c$.

To sum up we have the following solutions \\[1ex] %
\fbox{\begin{minipage}[t]{13cm}%
\begin{gather}
\xi_{A}{^{;A}}=-2G\cot x+\frac{2}{\sin x}\left(J_{0}\sin t+J_{1}\cos t\right)\label{k div}
\end{gather}
\begin{gather}
\xi_{A;B}\epsilon^{AB}=D\label{k rot xi}
\end{gather}
\begin{gather}
\xi_{3}=\left(a_{0}\sin t+a_{1}\cos t\right)\sin x+G-\left(J_{0}\sin t
+J_{1}\cos t\right)\cos x \label{k xi3}
\end{gather}
\begin{gather}
\xi_{0}=\left(a_{0}\cos t-a_{1}\sin t\right)\cos x+\left(J_{0}\cos t-J_{1}\sin t\right)\sin x+c\label{k xi 0}
\end{gather}
\end{minipage}} \\[1ex]

Here $a_{0},a_{1},c$ are constants,   
whereas $G,D,J_{0},J_{1}$ are dipole functions independent from
$x$ and $t$.\\
 The space of solutions has dimension 15, which is exactly the
number that was expected. 

We can now write the basis of the space of all the conformal Killing one-forms. According to
the equation \ref{eq:hodgea kodairy} we have
\begin{gather}
\xi_{A}=\overset{1}{v}_{;A}+\epsilon_{A}{^{B}}\overset{2}{v}_{;B}\,.
\end{gather}
Functions $\overset{1}{v}$ and $\overset{2}{v}$ may be calculated
using previously derived formulas
\[
\xi_{A}{^{;A}}=\overset{1}{v}_{;A}{^{A}}=-\frac{2}{\sin^2 x}\overset{1}{v}
\]
and
\begin{gather*}
\xi_{A;C}\epsilon^{AC}=\epsilon_{A}{^{B}}\overset{2}{v}_{;BC}\epsilon^{AC}
=\overset{2}{v}_{;A}{^{A}}=-\frac{2}{\sin^2 x}\overset{2}{v}\,.
\end{gather*}
We therefore have the following (linearly independent) conformal Killing one-forms for the metric $g$ from the equation \ref{konfrow}.
\[
R=-\frac{1}{2}\sin^2 x \, \epsilon{_{A}}{^{B}} D{_{,B}} \, \dd x{^{A}}
\]
\[
P=G\dd x+\sin x\cos x \, \dd G
\]
\[
T=c\dd t
\]
\[
B_{J_{0}}=-\sin x\sin t{\rm d}J{_{0}}-J{_{0}}\sin t\cos x\dd x+J{_{0}}\cos t\sin x\dd t
\]
\[
D=a_{1}\cos t\sin x\dd x-a_{1}\sin t\cos x\dd t
\]
\[
K{_{t}}=a_{0}\sin t\sin x\dd x+a_{0}\left(\cos t\cos x-1\right)\dd t
\]
\[
K_{J_{1}}=\left(\sin x\cos x-\sin x\cos t\right)\dd J{_{1}}-J_{1}\left(\cos t\cos x-1\right)\dd x-J_{1}\sin t\sin x\dd t
\]

From those conformal Killing one-forms for the metric $g$ from the equation \ref{konfrow} we can easily obtain conformal Killing one-forms for the 4 dimensional Anti-de Sitter metric from the equation \ref{prawdziwa} by multiplying them by the conformal factor $\frac 1 {\cos^2 x}$.

Close to $x=0$ our metric $g$ from the equation \ref{konfrow} is similar to Minkowski metric.
We are therefore expecting that for $x\rightarrow 0$ our Killing forms
will look similarly to the known conformal one-forms in the Minkowski
space-time. This turns out to be true. We see, that $R$ corresponds
to the generators of rotations in the Minkowski space-time, $P$ corresponds
to spacial translations, $T$ corresponds to time translation, $B{_{J{_{0}}}}$
to boosts, $D$ to dilation, $K{_{t}}$ to time acceleration, whereas
$K{_{J{_{1}}}}$ to space accelerations.



\section{Additional proofs}

In this appendix we will present proofs for some of the theorems used in this paper.

Let's prove theorem \ref{pochodnajednoformyyy}. It states that if
$\omega$ is a one-form on the manifold $N$ then
\[
\omega_{b|a}=\omega_{b;a}-K_{ab}\omega_{\mu}n^{\mu}\,.
\]

\begin{proof}
	We can assume that
	\begin{equation}
	\tilde{K}\left(X,Y\right)=K\left(X,Y\right)n=K_{ac}X^{a}Y^{c}n\,.
	\end{equation}
	We have then
	\begin{equation}
	\overset{N}{\nabla}_{a}v^{b}=\overset{M}{\nabla}_{a}v^{b}+K_{ac}v^{c}n^{b}\,.\label{eq:pochweknab}
	\end{equation}
	Using $|$ and $;$ (in the convention of the section \ref{submanifoldy} )we obtain
	\begin{align}
	v^{b}{_{|a}} & =v^{b}{_{;a}}+K_{ac}v^{c}n^{b} \, , \label{eq:poch kow wek}\\
	v^{n+1}{_{|a}} & =K_{ac}v^{c}n^{n+1}\,.
	\end{align}
	Here $v$ is tangent to $M$. We later have
	\begin{equation}
	 \left(v^{b}\omega_{b}\right)_{|a}=\left(v^{\mu}\omega_{\mu}\right)_{|a}
 =v^{\mu}{_{|a}}\omega_{\mu}+v^{\mu}\omega_{\mu|a}
  =v^{b}{_{;a}}\omega_{b}+K_{ac}v^{c}n^{\mu}\omega_{\mu}+v^{b}\omega_{b|a} \, .
	\end{equation}
	But on the other hand
	\begin{equation}
	 \left(v^{b}\omega_{b}\right)_{|a}=\left(v^{b}\omega_{b}\right)_{;a}
 =v^{b}{_{;a}}\omega_{b}+v^{b}\omega_{b;a}\, ,
	\end{equation}
	hence we get the result
	\begin{equation}
	\omega_{b|a}=\omega_{b;a}-K_{ab}\omega_{\mu}n^{\mu}\,.\label{eq:poch kow kowektor}
	\end{equation}
\end{proof}

Now we will prove the theorem \ref{krzywizna zewnetrzna}. It states that the external curvature form $K_{ab}$
satisfies equation $K=-\frac{1}{2}\mathcal{L}_{n}g$, where $n$ is a normal normalized field.

\begin{proof}
	Let $X$ and $Y$ be vector fields tangent to $M$. Now
	\begin{align*}
	\left(\mathcal{L}_{n}g\right)\left(X,Y\right) & =\mathcal{L}_{n}\left(g(X,Y)\right)-g(\mathcal{L}_{n}X,Y)-g(X,\mathcal{L}_{n}Y)\\
	& =\overset{N}{\nabla}_{n}\left(g\left(X,Y\right)\right)-g(\left[n,X\right],Y)-g(\left[n,Y\right],X)\\
	& =g(\overset{N}{\nabla}_{n}X,Y)+g(\overset{N}{\nabla}_{n}Y,X)-g(\overset{N}{\nabla}_{n}X-\overset{N}{\nabla}_{X}n,Y)-g(\overset{N}{\nabla}_{n}Y-\overset{N}{\nabla}_{Y}n,X)\\
	& =-2K(X,Y) \, ,
	\end{align*}
	where in the last step we used the following equation
	\begin{equation}
	 0=\overset{N}{\nabla}_{X}\left(g\left(n,Y\right)\right)=g(\overset{N}{\nabla}_{X}n,Y)+g(n,\overset{N}{\nabla}_{n}Y)\,.
	\end{equation}
\end{proof}

Now we will present the proof of the theorem \ref{dualofcyktensor}. It states that
Hodge dual of the CYK three-form is a CYK tensor.

\begin{proof}
	Let's define $s=\sgn\left(\det g\right)$. In the following calculations
	we will be using standard notations for symmetrization and skew-symmetrization i.e. $\alpha_{\left(ab\right)}:=\frac{1}{2}\left(\alpha_{ab}+\alpha_{ba}\right)$
	and $\alpha_{\left[ab\right]}:=\frac{1}{2}\left(\alpha_{ab}-\alpha_{ba}\right)$.
	Hodge dual we define as $\ast T_{ef}=\frac{1}{3!}\epsilon_{ef}{^{abc}}T_{abc}$
	. That is why we contracted \ref{eq:cyk trojforma} with tensor $\frac{1}{6}\epsilon_{ef}{^{abc}}$,
	where $\epsilon$ is a metric volume form of $M$. Additionally we
	denote $\chi_{f}=\ast T_{f}{^{c}}{_{; c}}$. We end up with some
	identities.
	\[
	2T_{ab\left(c;d\right)}\frac{1}{6}\epsilon_{ef}{^{abc}}=\ast T_{ef;d}+\frac{1}{6}\epsilon_{ef}{^{abc}}T_{abd;c}
	\]
	Let us evaluate $\frac{1}{6}\epsilon_{ef}{^{abc}}T_{abd;c}$. To this
	end we remind ourselves that for the $k$-form on $n$ dimensional
	manifold we have $\ast\ast\alpha=\left(-1\right)^{k\left(n-k\right)+s}\alpha$,
	where $s:=\sgn\left(\det g\right)$.
	We have
	\begin{align*}
	2\epsilon_{ef}{^{abc}}T_{abd;c} & =2s\epsilon_{ef}{^{abc}}\ast\ast T_{abd;c}=2s\epsilon_{ef}{^{abc}}\frac{1}{2}\epsilon_{abd}{^{kh}}\ast\!T_{kh;c}\\
	& =s\epsilon_{ef}{^{abc}}\epsilon_{abd}{^{kh}}\ast\!T_{kh;c}=s\ast\!T_{kh;c}\epsilon_{abef\tilde{c}}\epsilon^{ab\tilde{d}kh}g^{\tilde{c}c}g_{\tilde{d}d}\\
	& =\ast T_{kh;c}g^{c\tilde{c}}g_{\tilde{d}d}2\delta^{\tilde{d}kh}{_{ef\tilde{c}}}\\
	& \hspace*{-1.5cm}=2\left(\ast T_{fh;c}g^{hc}g_{ed}+\ast T_{ke;c}g^{kc}g_{fd}+\ast T_{ef;d}-\ast T_{eh;c}g^{hc}g_{fd}-\ast T_{fe;d}-\ast T_{kf;c}g^{kc}g_{ed}\right)\\
	& =2\left(\chi_{f}g_{ed}-\chi_{e}g_{fd}+2\ast\!T_{ef;d}-\chi_{e}g_{fd}+\chi_{f}g_{ed}\right)\\
	& =8\chi_{[f}g_{e]d}+4\ast\!T_{ef;d} \, ,
	\end{align*}
	where
	\begin{equation}
	\delta^{a_{1}a_{2}a_{3}}{_{b_{1}b_{2}b_{3}}}=\sum_{\pi\in S\left(3\right)}^ {}\sgn\left(\pi\right)\text{}\prod_{i\in\{1,2,3\}}^ {}\delta^{a_{\pi\left(i\right)}}{_{b_{i}}} \, .
	\end{equation}
	So
	\begin{equation}
	2T_{ab\left(c;d\right)}\frac{1}{6}\epsilon_{ef}{^{abc}} =
	\frac{2}{3}\chi_{[f}g_{e]d}+\frac{4}{3}\ast\!T_{ef;d} \, .
	\end{equation}
	But from the equality \ref{eq:cyk trojforma} it follows that
	\begin{gather}
	2T_{ab\left(c;d\right)}\frac{1}{6}\epsilon_{ef}{^{abc}}
	=\frac{1}{6}\epsilon_{ef}{^{abc}}\left(-2Q_{[ab}g_{c]d}+Q_{[ac}g_{b]d}-Q_{[bc}g_{a]d}\right) \nonumber \\
	=\frac{1}{6}\epsilon_{ef}{^{abc}}\left(-4\right)Q_{ab}g_{cd}=-4\ast\!Q_{efd} \, ,
	\end{gather}
	so
	\begin{equation}
	\ast T_{ef;d}=-3\ast\!Q_{efd}+\frac{1}{2}\chi_{[e}g_{f]d} \, .
	\end{equation}
	We check that
	\begin{equation}
	 2\ast\!T_{e\left(f;d\right)}=\frac{1}{4}\left(\chi_{e}g_{fd}-\chi_{f}g_{ed}+\chi_{e}g_{df}-\chi_{d}g_{ef}\right)=\frac{1}{4}\left(2\chi_{e}g_{fd}-\chi_{f}g_{ed}-\chi_{d}g_{ef}\right)\,,
	\end{equation}
	which is an equation satisfied by CYK tensor.
\end{proof}

%

\end{document}